\documentclass[11pt, leqno]{article}

\usepackage{amsmath, amssymb, amsfonts, stmaryrd, dsfont, mathrsfs, theorem, vmargin}

\setmarginsrb{26mm}{22mm}{26mm}{17mm}{0pt}{0pt}{0pt}{26pt}

\theoremheaderfont\scshape
\newtheorem{theorem}{Theorem}[section]
\newtheorem{lemma}[theorem]{Lemma}
\newtheorem{corollary}[theorem]{Corollary}
\newtheorem{definition}[theorem]{Definition}
\newtheorem{proposition}[theorem]{Proposition}

\newtheorem{assumption}[theorem]{Assumption}

\theorembodyfont{\rmfamily}
\newtheorem{remark}[theorem]{Remark}
\newtheorem{remarks}[theorem]{Remarks}

\newcommand\NN{\mathbb N}
\newcommand\PP{\mathbb P}
\newcommand\QQ{\mathbb Q}
\newcommand\RR{\mathbb R}
\newcommand\UU{\mathbb U}
\newcommand\VV{\mathbb V}

\newcommand\C{\mathscr C}
\newcommand\D{\mathscr D}
\newcommand\F{\mathscr F}

\renewcommand\H{\mathscr H}

\newcommand\M{\mathcal M}

\renewcommand\d{\mathrm d}

\newcommand\E[2][\PP]{\operatorname E_{#1}\left[ #2\right]}
\newcommand\bigE[2][\PP]{\operatorname E_{#1}\big[ #2\big]}

\newcommand\qed{\hfill$\Box$}
\newcommand\ip[2]{#2(#1)}
\newcommand\radnik[1][\QQ]{\frac{\d #1}{\d\PP}}
\newcommand\ind{\mathds1}
\newcommand\edow{\mathscr E}
\newcommand\Cone{\operatorname{Cone}}
\newcommand\CM{\Xi_\tau}
\newcommand\ba{{\operatorname{\textit{ba}}}}
\newcommand\optQQ{\widehat\QQ}

\newcommand\adm{\textup{adm}}
\newcommand\perm{\textup{perm}}
\newcommand\mg{\textup{mg}}

\newcommand\Kperm{K^\perm}
\newcommand\normcl[2][\QQ]{\operatorname{cl}_{L^1(#1)}(#2)}
\newcommand\esssup{\operatorname{ess\:sup}}
\newcommand\essinf{\operatorname{ess\:inf}}

\newcommand\essmax{\operatorname{ess\:max}}
\newcommand\lint\llbracket
\newcommand\rint\rrbracket
\newcommand\coloneqq{\mathrel{\mathop:}=}

\numberwithin{equation}{section}

\newenvironment{proof}[1][]{\noindent\textit{Proof#1.} }{\vskip\baselineskip}

\allowdisplaybreaks

\begin{document}
\title{Optimal Investment with an Unbounded Random Endowment and Utility-Based Pricing}

\author{Mark P. Owen\thanks{Maxwell Institute for Mathematical Sciences, and Department of Actuarial Mathematics and Statistics, Heriot-Watt University.}\and Gordan \v{Z}itkovi\'{c}\thanks{Department of Mathematics, University of Texas at
Austin.}} \date{}
\maketitle
\renewcommand{\abstractname}{}
\begin{abstract}\renewcommand\thefootnote{}\footnotesize
This paper studies the problem of maximizing the expected utility of terminal wealth for a financial agent with an unbounded random endowment, and with a utility function which supports both positive and negative wealth. We prove the existence of an optimal trading strategy within a class of permissible strategies -- those strategies whose wealth process is a supermartingale under all pricing measures with finite relative entropy. We give necessary and sufficient conditions for the absence of utility-based arbitrage, and for the existence of a solution to the primal problem.

We consider two utility-based methods which can be used to price contingent claims. Firstly we investigate marginal utility-based price processes (MUBPP's). We show that such processes can be characterized as local martingales under the normalized optimal dual measure for the utility maximizing investor. Finally, we present some new results on utility indifference prices, including continuity properties and volume asymptotics for the case of a general utility function, unbounded endowment and unbounded contingent claims.
\footnote{We thank Pavel Grigoriev for discussions on the topic of convex risk measures. We would also like to thank an anonymous referee for helpful comments. Mark Owen gratefully acknowledges support from EPSRC grant GR/S80202}
\end{abstract}

\section{Introduction}

The problems of optimal investment and contingent claim pricing in incomplete markets are fundamental in mathematical finance. The purpose of this article is to develop further the theory for the case of an unbounded random endowment and a utility function defined on the whole real line. We also investigate some of the consequences for utility-based approaches to pricing contingent claims.

The existing literature on optimal investment is far too broad for us to give a meaningful overview, so we concentrate on the area of immediate interest. The first paper to treat the case of general utility, within the framework of general semimartingale models of incomplete markets, was Kramkov and Schachermayer (1999). In this paper, the authors considered an investor, endowed with a deterministic initial wealth, with a utility function which supports positive wealths. The investor's objective is to maximize the expected utility of their terminal wealth by trading using admissible strategies -- those strategies whose wealth process is uniformly bounded from below by a constant. The authors demonstrate the existence of an optimal admissible trading strategy by first solving a dual optimization problem.

The case of a general utility function supporting negative wealth was solved by Schachermayer (2001) for the case of a locally bounded semimartingale. For the negative wealth case, it quickly becomes apparent that the class of admissible strategies is too narrow. However, one must somehow rule out arbitrage strategies, and a subsequent paper by Schachermayer (2003) showed that the wealth process of the optimal trading strategy indeed satisfies an important supermartingale property. Recently, Biagini and Frittelli (2005) have generalized the results of Schachermayer (2001) in an elegant way to the non-locally bounded case.

In Cvitani\'c, Schachermayer and Wang (2001) the optimal investment problem considered in Kramkov and Schachermayer (1999) was generalized by allowing a bounded random endowment, rather than just a deterministic initial wealth. The authors employed the duality between $L^\infty$ and finitely additive signed measures to treat the case of a bounded random endowment. The case of a bounded random endowment with a utility function supporting negative wealth was treated in Owen (2002).

For an exponential utility function Delbaen et al. (2002) solve the case of hedging a short position in a contingent claim which is unbounded above, by making a change of measure. Kabanov and Stricker (2002) have presented some additional and interesting results in this context. Hugonnier and Kramkov (2004) consider an unbounded endowment (whose modulus is super-hedgeable) for a general utility function which is defined on the positive real line. Meanwhile, \v{Z}itkovi\'c (2005) has presented the solution to the problem of optimal investment and consumption of wealth for an investor with an endowment which is unbounded above.

In the current article we assume that the investor has an unbounded random endowment which is assumed to be both super- and sub-hedgeable (for precise details, see Assumptions \ref{ass:edow} and \ref{ass:newweakerendowment}). This includes, of course, the known cases of a bounded random endowment, or simply just a deterministic initial wealth. We assume that the investor has a utility function supporting both positive and negative wealth, the only restriction being the now standard hypothesis of ``Reasonable Asymptotic Elasticity'' (see Assumption \ref{ass:rae} for more details). The agent's objective is to maximize their expected utility of terminal wealth derived from trading and the endowment, over the time interval $[0,T]$, where $T$ is some finite time horizon.

We treat, for simplicity, the case of a locally bounded semimartingale market model. For the case of a non-locally bounded semimartingale with a deterministic endowment, the reader is referred to the work of Biagini and Frittelli (2005, 2006). An analysis of the non-locally bounded case requires the use of carefully selected weight functions which bound from below the acceptable losses of a trading strategy.

Rather than treating the supermartingale property of the optimal wealth process as an added extra, we have taken the step of embedding this property directly in the formulation of the problem. We consider therefore the space of ``permissible'' trading strategies -- those strategies whose wealth process is a supermartingale under the pricing measures with finite relative entropy (see Definition \ref{def:perm}).

In Theorem \ref{thm:main} we show the existence of a unique solution to both the primal and dual problems, for an endowment which is unbounded above. By a simple observation, we extend these results to the case of an endowment which is sub- and super-hedgeable. In Theorem \ref{thm:minimalassumption} we complete the picture by showing that a necessary and sufficient condition for the well-posedness of the primal problem is the existence of an absolutely continuous local martingale measure with finite relative entropy. In addition, we show that a necessary and sufficient condition for the existence of an optimal solution to the primal problem is the existence of an \emph{equivalent} local martingale measure with finite relative entropy. These observations are new, to the best of our knowledge.

The proof of Theorem \ref{thm:main} is broken down into Sections \ref{sec:twwwwooooo}-\ref{sec:opttrading}. In Section \ref{sec:twwwwooooo} we consider the optimization problem on the whole interval $[0,T]$. Lemma \ref{thm:kabastri} presents a slight development of a result of Kabanov and Stricker which shows that any optimal measure in the dual problem will be ``as close to being equivalent to the statistical measure as possible'' within the set of local martingale measures with finite relative entropy. We show, in Proposition \ref{thm:lagrange_duality}, both that there is no duality gap and that there is an optimal measure in the dual problem by applying the Lagrange Duality Theorem.

In order to facilitate the proof of the supermartingale property we consider, in Section \ref{sec:dyndual}, a dynamic version of the dual problem, similar to that used by Schachermayer (2003). The dynamic dual problem is initialized at a stopping time $\tau$, valued in $[0,T]$, with the conditional expectation of the density of a local-martingale measure.

In Section \ref{sec:opttrading} we find the optimal trading strategy in the primal problem by showing that the optimal candidate terminal wealth is superhedgeable with zero initial endowment and applying results of Delbaen and Schachermayer in a fairly standard fashion. We then give a proof of the supermartingale property for the optimal wealth process, along the lines of Schachermayer (2003). Before we present the proof however, we give a new and, we believe, insightful proof of the supermartingale property in the case of exponential utility (see Proposition \ref{thm:exponly}), by showing that the optimal wealth process can be written as a Snell envelope. This idea is new, to the best of our knowledge, even for the case of deterministic initial wealth. The supermartingale property for general utility is given in Theorem \ref{thm:supermart}.

Section \ref{sec:dependow} contains an investigation of the properties of the optimal expected utility as a function of the endowment. This is worthwhile in its own right, but also helps us prove our results in Section \ref{sec:indiff}.

In the last two sections of the paper we turn to the topic of pricing contingent claims. A brief overview of some of the relevant literature can be found at the beginning of each section. In Section \ref{sec:mubpp} we introduce Marginal Utility-Based Price Processes (MUBPP's), which are essentially a dynamic version of the Marginal Utility-Based Price investigated by Hugonnier, Kramkov and Schachermayer (2005). The purpose of this section is to determine a fair price process for one or more new assets which are to be introduced into an existing economy. The dynamics of the asset prices should be consistent with the balance of supply and demand expected to hold for a financial market in equilibrium. If the new assets are fairly priced then the agent's expected utility of terminal wealth will not increase. Put differently, the agent's optimal demand for the new assets will be equal to zero. This idea is common in economics, where one considers a purely financial asset (in zero total supply) and a representative investor for the economy as a whole. In Theorem \ref{thm:MUBPP} we show that a locally bounded semimartingale is a MUBPP if and only if it is a local martingale under the (normalized) optimal dual measure.

In Section \ref{sec:indiff} we consider the implication of our results in terms of utility indifference pricing. Until now, most results on indifference prices have been proved for special cases of the utility function. See for example Delbaen et al. (2002) and Becherer (2003). In Proposition \ref{thm:indifprice} we consider properties of the indifference price for a general utility function. It is natural for us formulate the indifference price for an investor with a random endowment, although our results are new even for the case of a deterministic initial wealth. We show that indifference pricing can be considered as pricing with an entropic penalty, thus generalizing a representation of Delbaen et al. (2002) for any utility function. We also obtain various continuity properties which are usually associated with convex measures of risk. In Proposition \ref{thm:avindifprice} we generalize results of Delbaen et al. (2002) and Becherer (2003) on volume asymptotics of the indifference price, both as the volume increases to infinity and decreases to zero.

\vskip\baselineskip

Throughout the article, $U:\RR\rightarrow\RR$ will denote a strictly increasing, strictly concave, continuously differentiable utility function which satisfies the Inada conditions $U'(-\infty)\coloneqq\lim_{x\rightarrow-\infty}U'(x)=\infty$ and $U'(\infty)\coloneqq\lim_{x\rightarrow\infty}U'(x)=0$. In order to avoid pathological examples, we assume throughout the hypothesis of Reasonable Asymptotic Elasticity, namely that
\begin{assumption}\label{ass:rae}
  \begin{equation*}
    \operatorname{AE}_{-\infty}(U)\coloneqq\liminf_{x\rightarrow-\infty}\frac{xU'(x)}{U(x)}>1,
    \qquad
    \operatorname{AE}_{+\infty}(U)\coloneqq\limsup_{x\rightarrow\infty}\frac{xU'(x)}{U(x)}<1.
  \end{equation*}
\end{assumption}
This assumption essentially ensures that the investor is still reasonably risk averse when they are either very wealthy or very poor (see Kramkov and Schachermayer (1999) and Schachermayer (2001) for a detailed discussion of the two conditions).

We extend the domain of $U$ in a natural way to $[-\infty,\infty]$ by prescribing that $U(-\infty)\coloneqq-\infty$ and $U(\infty)\coloneqq\sup_{x\in\RR}U(x)$. The convex conjugate $V:(0,\infty)\rightarrow\RR$ of $U$ is defined by
\[ V(y)\coloneqq\sup_{x\in\RR}\{U(x)-xy\}. \]
It is an elementary exercise in convex analysis to show that $V$ is a strictly convex, continuously differentiable function and that $(U')^{-1}=-V'$. We may therefore extend the domains of the functions $V$ and $V'$ in a natural way to $[0,\infty]$, by prescribing that
\[ V(0)\coloneqq U(\infty),\quad V(\infty)\coloneqq\infty, \quad V'(0)\coloneqq-\infty \quad \text{and} \quad V'(\infty)\coloneqq\infty. \]
Provided $U(0)>0$, which we henceforth assume with no loss of generality, it turns out that Assumption \ref{ass:rae} (Reasonable Asymptotic Elasticity) is equivalent to the assumption of either of the following two growth conditions on $V$.
\newcounter{SavedEnum}
\begin{assumption}\ \label{ass:growth}
  \begin{enumerate}
    \item For any $\lambda>0$ there exists a constant $C>0$ such that $V(\lambda y)\le CV(y)$, for all $y\ge0$,
    \setcounter{SavedEnum}{\value{enumi}}
  \end{enumerate}
  or, equivalently,
  \begin{enumerate}
    \setcounter{enumi}{\value{SavedEnum}}
    \item there exists a constant $C'>0$ such that $y|V'(y)|\le C'V(y)$ for all $y>0$.
  \end{enumerate}
\end{assumption}
If $U(0)\le0$ then we would need to use a slightly more cumbersome version of the growth conditions above, however all our results are clearly still valid. For a proof of the equivalence of Assumptions \ref{ass:rae} and \ref{ass:growth}, we refer the reader to Proposition 4.1 and Corollary 4.2 of Schachermayer (2001) and Lemma 6.3 of Kramkov and Schachermayer (1999).

\vskip\baselineskip

We model the discounted prices of $d$ risky assets by a locally bounded, $\RR^d$-valued semimartingale, $(S_t)_{t\in[0,T]}$, defined on a filtered probability space $(\Omega,\F,(\F_t)_{t\in[0,T]},\PP)$, in which the filtration is assumed to satisfy the usual conditions of right continuity and completeness.

Let $\M^a$ (resp. $\M^e$) denote the set of probability measures $\QQ\ll\PP$ (resp. $\QQ\sim\PP$) such that $S$ is a local martingale under $\QQ$. To avoid ambiguity in Section \ref{sec:mubpp} we shall write $\M^a(S)$ for the set of local martingale measures for the process $S$.

For a measure $\mu\ll\PP$, the quantity $\E{V\left(\radnik[\mu]\right)}\in (0,\infty]$ is known as the (generalized) entropy of $\mu$ relative to $\PP$. Let $\M_V^a$ (resp. $\M_V^e$) denote the convex set of probability measures $\QQ\in\M^a$ (resp. $\QQ\in\M^e$) with finite relative entropy. Note that if $V(0)=U(\infty)=\infty$ then $\M_V^a=\M_V^e$, but otherwise this may not be the case.

Throughout the paper we shall use measures with non-unit total mass. We introduce therefore the notation
\[ \Cone(\M)\coloneqq\{y\QQ:y\ge0,\QQ\in\M\}, \]
where $\M$ is any convex collection of measures. Note that all non-zero measures in $\Cone(\M_V^a)$ have finite relative entropy due to Assumption \ref{ass:growth}(i).

For $\PP'\ll\PP$, let $L(S;\PP')$ denote the space of $\RR^d$-valued, predictable, $S$-integrable processes under $\PP'$. In order to rule out doubling strategies (whose ``upward creep'' leads to arbitrage), it is common practice in the literature to consider trading strategies which are admissible, in the sense that they belong to the cone
\[ \H^\adm\coloneqq\{H\in L(S;\PP):H\cdot S\text{ is }\PP\text{-a.s. uniformly bounded from below by a constant}\}. \]
As is well known however, this class lacks the necessary closure properties to admit a solution to the optimal investment problem when the utility function is defined on the whole real line. In Definition \ref{def:perm} below, we widen this class sufficiently so that it contains the optimal strategy.

The primal problem we shall formulate shortly assumes that as well as being able to trade on the financial market, our investor has a random endowment, which is represented by an $\F$-measurable random variable $\edow$.
\begin{assumption}\label{ass:edow}
  There exist $x',x''\in\RR$ and $H''\in\H^\adm$ such that
  \[ x'\le\edow\le x''+(H''\cdot S)_T. \]
\end{assumption}
Clearly any bounded random endowment, or even just a deterministic initial wealth satisfies Assumption \ref{ass:edow}. See Assumption \ref{ass:newweakerendowment} for a wider class of endowments, which can also be treated by a simple observation.

The primal problem for the investor is to maximize their expected utility of terminal wealth derived from trading and the endowment. In order that the primal problem can be sensibly formulated, it is important to rule out arbitrage. In particular, a minimum requirement is that
\begin{equation}\label{eqn:admprimalproblem}
  u_\edow^\adm\coloneqq\sup_{H\in\H^\adm}\E{U\left(\int_0^TH_u\d S_u+\edow\right)}<U(\infty).
\end{equation}
In Theorem \ref{thm:minimalassumption} we show that, under Assumption \ref{ass:edow}, the assumption \eqref{eqn:admprimalproblem} is equivalent to the following, which shall therefore hold throughout the paper.
\begin{assumption}\label{ass:noarb}
  $\M_V^a\neq\emptyset$.
\end{assumption}
We will indicate where Assumption \ref{ass:noarb} needs strengthening to the existence of an {\it equivalent} local martingale measure with finite entropy (for example in Theorem \ref{thm:main}).

As remarked above the cone $\H^\adm$ of trading strategies is not large enough for the optimal investment problem \eqref{eqn:admprimalproblem} to admit a solution. We consider therefore the following wider class of trading strategies.
\begin{definition}\label{def:perm}
  We shall say that $H$ is a \emph{permissible} trading strategy if it lies in the cone
  \begin{equation*}
    \H^\perm\coloneqq\{H\in L(S;\PP): H\cdot S\text{ is a }\QQ\text{-supermartingale for all }\QQ\in\M_V^a \}.
  \end{equation*}
\end{definition}
While $\H^\perm$ is itself not a vector space, it contains the space $\H^\mg\coloneqq\{H\in L(S;\PP): H\cdot S\text{ is a }\QQ\text{-martingale for all }\QQ\in\M_V^a \}$.

Using these wider classes of trading strategies, we can formulate a weaker version of Assumption \ref{ass:edow}, which we can also treat.

\begin{assumption}\label{ass:newweakerendowment}
  There exist $x',x''\in\RR$, and trading strategies $H'\in\H^\mg$ and $H''\in\H^\perm$, such that
  \[ x'+(H'\cdot S)_T\le\edow\le x''+(H''\cdot S)_T. \]
\end{assumption}
We may now formulate the primal and dual problems.
\begin{definition}\label{def:optproblems}
  \begin{enumerate}
    \item (Primal Problem)
    \begin{equation}\label{eqn:primalproblem}
      u_\edow\coloneqq\sup_{H\in\H^\perm}\E{U\left(\int_0^TH_u\d S_u+\edow\right)};
    \end{equation}
    \item (Dual Problem)
    \begin{equation}\label{eqn:maindual}
      v_\edow\coloneqq\inf_{\mu\in\Cone(\M_V^a)}\E{V\left(\radnik[\mu]\right)+\radnik[\mu]\edow}.
    \end{equation}
  \end{enumerate}
\end{definition}
Note that $u_\edow^\adm\le u_\edow$ because $\H^\adm\subseteq\H^\perm$ (see Theorem 2.9 of Delbaen and Schachermayer (1994)). Later on, in Theorem \ref{thm:minimalassumption}, we show that $u_\edow^\adm=u_\edow$. Note that $v_\edow=\inf_{y\ge0}v_\edow(y)$ where $v:[0,\infty)\rightarrow(-\infty,\infty]$ is defined by
\begin{equation}\label{eqn:simpledual}
  v_\edow(y)\coloneqq\inf_{\QQ\in\M_V^a}\E{V\left(y\radnik\right)+y\radnik\edow}.
\end{equation}
It follows from convexity of $V$ and $\M_V^a$ that $v_\edow(y)$ is convex as a function of $y$. Note also that $v_\edow(0)=V(0)=U(\infty)$.

In preparation for Theorem \ref{thm:main}, note that it follows immediately from the definition of $V$ that for any $H\in\H^\perm$ and $\mu\in\Cone(\M_V^a)$ we have
\begin{align}\label{eqn:duality}
  \E{U((H\cdot S)_T+\edow)} &\le \E{V\left(\radnik[\mu]\right)+\radnik[\mu]((H\cdot S)_T+\edow)}\notag\\
  &\le \E{V\left(\radnik[\mu]\right)+\radnik[\mu]\edow}.
\end{align}
Taking the supremum of the left hand side over all $H\in\H^\perm$ and the infimum of the right hand side over all $\mu\in\Cone(\M_V^a)$ gives
\begin{equation}\label{eqn:easypart}
  u_\edow\le v_\edow\le U(\infty).
\end{equation}
The following theorem is the first of several of our central results.
\begin{theorem}\label{thm:main}
  Suppose that the investor's utility function satisfies Assumption \ref{ass:rae}. Suppose also that Assumption \ref{ass:noarb} holds (i.e. $\M_V^a\neq\emptyset$), and that the random endowment satisfies either Assumption \ref{ass:edow} or the weaker Assumption \ref{ass:newweakerendowment}. Then
  \begin{enumerate}
    \item $u_\edow=v_\edow<U(\infty)$; and \item there exists a $\widehat\mu\in\Cone(\M_V^a)\setminus\{0\}$ which is optimal in the dual problem \eqref{eqn:maindual}.
  \end{enumerate}
  If, in addition, $\M_V^e\neq\emptyset$ then
  \begin{enumerate}\setcounter{enumi}2
    \item $\widehat\mu\in\Cone(\M_V^e)\setminus\{0\}$ and there exists an $\widehat H\in\H^\perm$ which is optimal in the primal problem \eqref{eqn:primalproblem}.
  \end{enumerate}
\end{theorem}

\begin{proof}
  Suppose first that the random endowment satisfies Assumption \ref{ass:edow}. In this case, parts (i) and (ii) of the above theorem are proved in Proposition \ref{thm:lagrange_duality}. If, in addition, $\M_V^e\neq\emptyset$ then $\widehat\mu\sim\PP$ by Lemma \ref{thm:kabastri}. Proposition \ref{thm:defoptgain} (see also Remark \ref{rem:aboutequiv}) and Theorem \ref{thm:supermart} provide a proof of the existence of an $\widehat H\in\H^\perm$  such that
  \[ U'\big((\widehat H\cdot S)_T+\edow\big)=\radnik[\widehat\mu] \qquad\text{or, equivalently,}\qquad (\widehat H\cdot S)_T=-V'\left(\radnik[\widehat\mu]\right)-\edow, \]
  and
  \[ \E[\widehat\mu]{\big(\widehat H\cdot S\big)_T}=0. \]
  Since, for such $\widehat\mu$ and $\widehat H$, the inequalities in \eqref{eqn:duality} become equalities, the proof is finished for this case.

  Suppose now that the random endowment satisfies Assumption \ref{ass:newweakerendowment}. Define $\widetilde\edow=\edow-(H'\cdot S)_T$. Then
  \[ x'\le\widetilde\edow\le x''+((H''-H')\cdot S)_T. \]
  Since $H''-H'\in\H^\perm$ it follows from the above inequalities that $(H''-H')\cdot S$ is uniformly bounded from below by the constant $x'-x''$, and thus $H''-H'\in\H^\adm$. Hence $\widetilde\edow$ satisfies the conditions of Assumption \ref{ass:edow}.

  By Proposition \ref{thm:lagrange_duality} there exists a unique $\widehat\mu\in\Cone(\M_V^a)\setminus\{0\}$ such that
  \[ v_{\widetilde\edow} = \E{V\left(\radnik[\widehat\mu]\right)+\radnik[\widehat\mu]\widetilde\edow}. \]
  Since $H'\cdot S$ is a martingale under all measures in $\M_V^a$ it follows that,
  \[ v_\edow=v_{\widetilde\edow}=\E{V\left(\radnik[\widehat\mu]\right)+\radnik[\widehat\mu]\widetilde\edow}
  =\E{V\left(\radnik[\widehat\mu]\right)+\radnik[\widehat\mu]\edow}, \]
  so $\widehat\mu$ is also a minimizer in the dual problem for the endowment $\edow$. It follows from strict convexity of $V$ that the minimizer is unique.

  If, in addition, $\M_V^e\neq\emptyset$ then our earlier results show that $\widehat\mu\in\Cone(\M_V^e)\setminus\{0\}$ and there exists an $\widetilde H\in\H^\perm$ such that
  \[ u_{\widetilde\edow}=\E{U\left((\widetilde H\cdot S)_T+\widetilde\edow\right)}. \]
  Define $\widehat H\coloneqq\widetilde H-H'\in\H^\perm$. Again, since $H'\cdot S$ is a martingale under all measures in $\M_V^a$ it follows that $u_\edow=u_{\widetilde\edow}=v_{\widetilde\edow}=v_\edow$ and
  \begin{equation*}
    u_\edow=u_{\widetilde\edow}=\E{U\big((\widetilde H\cdot S)_T+\widetilde\edow\big)}=
    \E{U\big((\widehat H\cdot S)_T+\edow\big)},
  \end{equation*}
  so $\widehat H\in\H^\perm$ is optimal in the primal problem.\qed
\end{proof}

In the next result, we show that both of our assumptions on the existence of local martingale measures are optimal. In the case when $\M_V^a\neq\emptyset$ but $\M_V^e=\emptyset$, the ``optimal terminal wealth'' will be infinite on the event $\{\radnik[\widehat\mu]=0\}$, which has non-zero measure under $\PP$. There is therefore no hope of representing the optimal terminal wealth as a stochastic integral. Recall that if $U(\infty)=\infty$ then the statements $\M_V^a\neq\emptyset$ and $\M_V^e\neq\emptyset$ are equivalent. One can interpret the following proposition as a version of the Fundamental Theorem of Asset Pricing - it relates the existence of a ``martingale'' measure to the notion of no-arbitrage phrased in terms of finiteness of the maximal utility.
\begin{theorem}\label{thm:minimalassumption}\
  \begin{enumerate}
    \item Under Assumptions \ref{ass:rae} and \ref{ass:edow}, Assumption \ref{ass:noarb} holds (i.e. $\M_V^a$ is non-empty) if and only if $u_\edow^\adm<U(\infty)$. In this case, $u_\edow^\adm=u_\edow$.

    \item Under Assumptions \ref{ass:rae}, \ref{ass:noarb} and \ref{ass:newweakerendowment},
    the set $\M_V^e$ is non-empty if and only if there exists an optimal $\widehat
    H\in\H^\perm$ in the primal problem.
  \end{enumerate}
\end{theorem}

\begin{proof}
  We postpone the proof of part (i) to Section 2. For part (ii), the ``only if'' direction is shown in Theorem \ref{thm:main}. To show the other direction, let $\widehat\mu$ denote the minimizer in the dual problem and suppose there exists an optimal $\widehat H\in\H^\perm$ in the primal problem. Then since, by Theorem \ref{thm:main}, there is no duality gap, and $(\widehat H\cdot S)_T+\edow<\infty$ a.s. we must have $\radnik[\widehat\mu]=U'((\widehat H\cdot S)_T+\edow)>0$ a.s., so $\widehat\mu\sim\PP$.\qed
\end{proof}

\subsubsection*{Finitely Additive Measure Theory}

We finish the introduction with some details concerning the theory of finitely additive signed measures (also known as finitely additive set functions), and refer the reader to \S20 of Hewitt and Stromberg (1965) or Chapter III of Dunford and Schwartz (1964) for further details. Let $\ba$ denote the vector space of all bounded, finitely additive, signed measures on $(\Omega,\F)$ which are absolutely continuous with respect to $\PP$. To be precise, $\ba$ consists of all real-valued functions $\mu$ defined on $\F$ which satisfy the conditions (i) $\sup\{|\mu(A)|:A\in\F\}<\infty$; (ii) $\mu(A)=0$ if $A\in\F$ and $\PP(A)=0$; and (iii) $\mu(A\cup B)=\mu(A)+\mu(B)$ if $A,B\in\F$ and $A\cap B=\emptyset$.

Let $\ba_+$ denote the set of finitely additive unsigned (i.e. non-negative) measures in $\ba$. Elements of $ba$ admit a Jordan decomposition $\mu=\mu_+-\mu_-$ where $\mu_+,\mu_-\in\ba_+$, which allows us to define the total variation of a finitely additive signed measure as $|\mu|=\mu_++\mu_-$.

The space $\ba$ can be endowed with the total variation norm $\|\mu\|\coloneqq|\mu|(\Omega)$. Condition (i) above ensures that the norm is finite. Of course, if $\mu\in\ba_+$ then $\|\mu\|=\mu(\Omega)$. When we discuss the normalization of an non-zero element $\mu\in\ba_+$, we shall mean the measure $\QQ\coloneqq\mu/\|\mu\|=\mu/\mu(\Omega)$ with unit mass.

Let $\ba^\sigma$ be the space consisting of those measures in $\ba$ which are sigma additive (also known in some texts as countably additive). Elements of $\ba^\sigma$ are therefore real-valued functions $\mu$ defined on $\F$ which satisfy condition (i), (ii) and the following strengthened condition (iii)' $\mu(\cup_{n=1}^\infty A_n)=\sum_{n=1}^\infty\mu(A_n)$ for all pairwise disjoint sequences $(A_n)_{n=1}^\infty$ of $\F$.

In addition to the usual theory of integration for measures in $\ba^\sigma$, one may also develop a theory of integration for measures in $\ba$ (see Section III.2 of Dunford and Schwartz (1964) or p.354ff in Hewitt and Stromberg (1965)). As a special case, the expression $\int_\Omega X\d\mu$ is well defined for any random variable $X\in L^\infty(\Omega,\F,\PP)$. For notational convenience, if $\mu\in\ba$ we shall also let $\mu$ denote the corresponding functional $\mu:L^\infty(\Omega,\F,\PP)\rightarrow\RR$ defined by $\ip X\mu\coloneqq\int_\Omega X\d\mu$. The intended meaning of this ``overloaded'' notation will be clear from the context.

It can be shown that the mapping $L^\infty\ni X\mapsto\mu(X)$ is continuous with respect to the $L^\infty$ norm topology. In fact, the correspondence between finitely additive signed measures and continuous linear functionals provides an isometric isomorphism between $(\ba,\|.\|)$ and $L^\infty(\Omega,\F,\PP)^*$ (see Theorem 20.35 of Hewitt and Stromberg (1965)). It follows therefore that $(\ba,\|.\|)$ is a Banach space, with a $\sigma(\ba,L^\infty)$-compact unit ball $\{\mu\in\ba:\|\mu\|\le1\}$. Meanwhile for signed measures $\mu\in\ba^\sigma$ the mapping $\mu\mapsto\radnik[\mu]$ provides an isometric isomorphism of $(\ba^\sigma,\|.\|)$ onto $L^1(\Omega,\F,\PP)$.

\section{A Global Constrained Optimization Problem}\label{sec:twwwwooooo}

In order to solve the constrained optimization problem \eqref{eqn:primalproblem}, it is convenient to pass to a related, abstract version of the problem which is formulated within $L^\infty$. This allows us to appeal to the duality theory between $L^\infty$ and $ba$. Throughout this section we assume that the random endowment $\edow$ satisfies Assumption \ref{ass:edow}. We also assume that Assumption \ref{ass:noarb} holds throughout, with the exception of Lemma \ref{thm:useful_tool} and the proof of Theorem \ref{thm:minimalassumption}(i), which is found at the end of the section.
\begin{definition}
  Define the concave functional $\UU_\edow:L^\infty\rightarrow\RR$ by
  \begin{equation}\label{eqn:defU}
    \UU_\edow(X)\coloneqq\E{U(X+\edow)}.
  \end{equation}
\end{definition}
Note that the expectation above exists and is finite because $U(X+\edow)$ is bounded below by a constant, and above by any of the integrable random variables $V\left(\radnik\right)+(X+\edow)\radnik\in L^1(\PP)$ with $\QQ\in \M_V^a$.

We now consider the abstract maximization problem
\begin{equation}\label{eqn:abstrmax} \sup_{G\in\C}\UU_{\edow}(G),\qquad\text{where}\qquad
  \C\coloneqq\{X\in L^\infty:X\le (H\cdot S)_T\text{ for some }H\in\H^\perm\}
\end{equation}
is the cone of bounded random variables which are dominated by the terminal gain of a permissible trading strategy. Note that it follows immediately from the definition of $\C$ that
\begin{equation}\label{eqn:dunno}
  \sup_{G\in\C}\UU_{\edow}(G)\le u_\edow.
\end{equation}
The convex functional $\VV_\edow:\ba\rightarrow(-\infty,\infty]$, conjugate to $\UU_\edow$, is defined by (c.f. \v{Z}itkovi\'c (2005))
\[ \VV_\edow(\mu)\coloneqq\sup_{X\in L^\infty}(\UU_\edow(X)-\ip X\mu). \]
The functional $\VV_\edow$ is lower semicontinuous with respect to the weak* topology $\sigma(\ba,L^\infty)$ because it is defined as the supremum of the $\sigma(\ba,L^\infty)$-continuous affine linear functionals $f_X:\ba\rightarrow\RR$ given by $f_X(\mu)\coloneqq\UU_\edow(X)-\ip X\mu$.

We define an abstract dual problem to \eqref{eqn:abstrmax} by
\begin{equation}\label{eqn:abstrglobdual}
  \inf_{\mu\in\D}\VV_\edow(\mu),\qquad\text{where}\qquad \D\coloneqq\{\mu\in \ba:\ip
  X{\mu}\le 0\text{ for all }X\in \C\}
\end{equation}
is the polar cone to $\C$. Note that $\D\subseteq\ba_+$ because $-L_+^\infty\subseteq \C$.

A well known route to solving the constrained optimization problem \eqref{eqn:abstrmax} is to introduce the Lagrangian $L(X,\mu)\coloneqq\UU_\edow(X)-\ip X\mu$, in which the measure $\mu$ acts as a Lagrange multiplier. Note that
\[ \sup_{X\in\C}\UU_{\edow}(X)\le\sup_{X\in L^\infty}\inf_{\mu\in\D}L(X,\mu)\le\inf_{\mu\in\D}\sup_{X\in L^\infty}L(X,\mu)=\inf_{\mu\in\D}\VV_\edow(\mu). \]

The following lemma probes deeper into the structure of the dual functional $\VV_\edow$ and relates the abstract approach to duality employed in this paper to the one used, for example, by Cvitani\'c, Schachermayer and Wang (2001).
\begin{lemma}\label{thm:useful_tool}
  If the endowment $\edow$ satisfies Assumption \ref{ass:edow} then
  \begin{equation*}
    \VV_\edow(\mu)=
    \begin{cases}\displaystyle\E{V\left(\radnik[\mu]\right)+\radnik[\mu]\edow},&\text{ if }\mu\in\ba^\sigma_+;\\
    \infty,&\text{ otherwise}.\end{cases}
  \end{equation*}
\end{lemma}

\begin{proof}
  Suppose first that $\mu\not\in\ba_+$. Then there exists $A\in\F$ such that
  $\mu(A)<0$. Hence $\VV_\edow(\mu)\ge\sup_{n\ge0}(\UU_\edow(n\ind_A)-n\mu(A))
  \ge\UU_\edow(0)-\mu(A)\sup_{n\ge0}n=\infty$.

  If $\mu\in\ba_+\setminus\ba^\sigma_+$ then it follows (from the definition of a sigma additive
  measure) that there exists a decreasing sequence $\{A_n\}_{n\in\NN}$ of events such that $\PP[A_n]\rightarrow0$ and $\inf_n\mu[A_n]>0$. Given a constant $k>0$, set $X_n^k=-k\ind_{A_n}$. Then $\UU_\edow(X_n^k)\nearrow\UU_\edow(0)$ as $n\rightarrow\infty$ by the monotone convergence theorem, so $\VV_\edow(\mu)\ge\sup_{k>0}\sup_n(\UU_\edow(X_n^k)-\ip{X_n^k}\mu) \ge \UU_\edow(0)+\sup_{k>0}k\inf_n\mu(A_n) = \infty$.

  For the remainder of the proof, we consider those $\mu\in\ba^\sigma_+$. Define $X=-V'\left(\radnik[\mu]\right)$. For $m,n\in\NN$, define $X_m=(-m)\vee(X\wedge m)\in L^\infty$, an $G_{m,n}\coloneqq X_m-\edow\ind_{\{\edow\le n\}}\in L^\infty$, so that $\UU_\edow(G_{m,n})\ge\E{U(X_m)}$. An application of the monotone convergence theorem gives $-\ip{G_{m,n}}\mu=\E[\mu]{\edow\ind_{\{\edow\le n\}}-X_m}\nearrow\E[\mu]{\edow-X_m}$ as $n\rightarrow\infty$. Define $V_m(y)\coloneqq\sup_{-m\le x\le m}(U(x)-xy)$. Clearly, $U(0)\le V_m(y)\nearrow V(y)$ as $m\rightarrow\infty$, and it is easy to show that the supremum is attained at $x=(-m)\vee((-V'(y))\wedge m)$. Therefore, by the monotone convergence theorem we have
  \begin{align*}
  \VV_\edow(\mu)&\ge\sup_m\sup_n(\UU_\edow(G_{m,n})-\ip {G_{m,n}}\mu)\ge\sup_m\E{U(X_m)-\radnik[\mu]X_m+\radnik[\mu]\edow} \\
  &=\sup_m\E{V_m\left(\radnik[\mu]\right)+\radnik[\mu]\edow}=\E{V\left(\radnik[\mu]\right)+\radnik[\mu]\edow}.
  \end{align*}
  For the opposite inequality take any $X\in L^\infty$. Then
  \[ \UU_\edow(X)-\ip X\mu=\E{U(X+\edow)-\radnik[\mu](X+\edow)+\radnik[\mu]\edow} \le \E{V\left(\radnik[\mu]\right)+\radnik[\mu]\edow} \]
  so $\VV_\edow(\mu)\le\E{V\left(\radnik[\mu]\right)+\radnik[\mu]\edow}$.\qed
\end{proof}

\begin{lemma}\label{thm:technical}
  \[ \Cone(\M_V^a)\subseteq\D\cap\ba^\sigma\subseteq\Cone(\M^a). \]
\end{lemma}

\begin{proof}
  Take any $\mu\in\Cone(\M_V^a)$. It follows from the definition of $\H^\perm$ that $\mu(G)\le0$ for all $G\in\C$, so $\mu\in\D$. The second inclusion follows from Theorem 5.6 of Delbaen and Schachermayer (1994).\qed
\end{proof}

We give next a slight development of Proposition 3.1 of Kabanov and Stricker (2002), which will be useful in the proofs of Propositions \ref{thm:lagrange_duality} and \ref{thm:suboptimal}.
\begin{lemma}\label{thm:kabastri}
  Suppose that Assumptions \ref{ass:rae}, \ref{ass:edow} and \ref{ass:noarb} hold. Let $\mathcal D\subseteq\D$ be a non-empty convex set such that $\VV_\edow(\mu)<\infty$ for some $\mu\in\mathcal D$. Suppose that $\widehat\mu\in\mathcal D$ is a minimizer for the problem $\inf_{\mu\in\mathcal D}\VV_\edow(\mu)$. Then $\mu\ll\widehat\mu$ for any $\mu\in\mathcal D$ such that $\VV_\edow(\mu)<\infty$. As a result, any minimizer must lie in $\Cone(\M_V^a)\setminus\{0\}$. Furthermore, if $\M_V^e$ is non-empty then $\widehat\mu\sim\PP$.
\end{lemma}

\begin{proof}
  Let $\widehat\mu$ be a minimizer, and take any $\mu$ such that $\VV_\edow(\mu)<\infty$. By Lemma \ref{thm:useful_tool}, both $\mu$ and $\widehat\mu$ must be sigma additive and have finite relative entropy. It follows from Lemma \ref{thm:technical} that $\mu,\widehat\mu\in\Cone(\M_V^a)$.

  Suppose for a contradiction there exists an event $A\in\F$ such that $\widehat\mu(A)=0$ and $\mu(A)>0$. For $\lambda\in[0,1]$, define $\mu_\lambda\coloneqq\lambda\mu+(1-\lambda)\widehat\mu\in\mathcal D\cap\Cone(\M_V^a)$ and $\nu_\lambda\coloneqq V\left(\radnik[\mu_\lambda]\right)+\radnik[\mu_\lambda]\edow$. Since $\lambda\mapsto\nu_\lambda$ is convex, the integrable random variables $(\nu_\lambda-\nu_0)/\lambda$ are monotone increasing in $\lambda$. By the monotone convergence theorem,
  \[ \lim_{\lambda\searrow0}\E{\ind_A\left(\frac{\nu_\lambda-\nu_0}\lambda\right)}=\E{\ind_A\lim_{\lambda\searrow0}\left(\frac{\nu_\lambda-\nu_0}\lambda\right)}=\E{\ind_A\radnik[\mu](V'(0)+\edow)}=-\infty. \]
  Hence $\lim_{\lambda\searrow0}\frac1\lambda\E{\nu_\lambda-\nu_0}=-\infty$. However, optimality of $\widehat\mu$ implies that $\E{\nu_\lambda-\nu_0}=\VV_\edow(\mu_\lambda)-\VV_\edow(\widehat\mu)\ge0$ for all
  $\lambda$, which is the required contradiction. The two final statements follow immediately.\qed
\end{proof}
While the abstract maximization problem \eqref{eqn:abstrmax} does not generally have an optimal solution $\widehat G\in\C$, the dual problem $\inf_{\mu\in\D}\VV_\edow(\mu)$ does have an optimal solution $\widehat\mu\in\D$ which turns out to be the measure in Theorem \ref{thm:main}. The following proposition demonstrates the existence of the optimal measure in the dual problem by employing the Lagrange Duality Theorem. The use of this powerful and elegant theorem (which is a version of the Separating Hyperplane Theorem) replaces the usual application of the well known minimax theorem.
\begin{proposition}\label{thm:lagrange_duality}
  Under Assumptions \ref{ass:rae}, \ref{ass:edow} and \ref{ass:noarb} there exists a unique minimizer $\widehat\mu\in\Cone(\M_V^a)\setminus\{0\}$ in the dual problems \eqref{eqn:maindual} and \eqref{eqn:abstrglobdual}, and there is no duality gap in \eqref{eqn:easypart}. Indeed,
  \begin{equation}\label{eqn:partofit}
    \sup_{G\in\C}\UU_\edow(G)=u_\edow=v_\edow=\min_{\mu\in\D}\VV_\edow(\mu)<U(\infty).
  \end{equation}
\end{proposition}

\begin{proof}
  From \eqref{eqn:dunno}, \eqref{eqn:easypart}, Lemmas \ref{thm:useful_tool} and \ref{thm:technical} we have
  \begin{equation*}
    \sup_{G\in\C}\UU_{\edow}(G)\le u_\edow\le v_\edow=\inf_{\mu\in\D}\VV_\edow(\mu)<\infty.
  \end{equation*}
  The result now follows simply by applying the Lagrange Duality Theorem: Following the notation of Luenberger (1969), \S8, we set $X=Z=L^\infty$, $\Omega=X$, and let $G:X\rightarrow Z$ be the identity operator. Let $P\coloneqq-\C$ be the positive cone of $Z$, so that the positive cone $P^\oplus$ of $Z^*$ is $\D$. Note that, as a subset of $Z$, the cone $P$ contains an interior point, namely the random variable $1$. It is easy to verify that the concave dual of the convex functional $f\coloneqq-\UU_{\edow}$ is $\phi(\mu)=-\VV_\edow(\mu)$. Applying Theorem 8.6.1 of Luenberger (1969) gives
  \begin{equation*}
    \sup_{C\in \C}\UU_\edow(C)=-\inf_{G(C)\le0}f(C)=-\max_{\mu\ge0}\phi(\mu)=\min_{\mu\in\D}\VV_\edow(\mu),
  \end{equation*}
  and that the optima on the right are achieved by some $\widehat\mu\in\D$. Applying Lemma \ref{thm:kabastri}, with $\mathcal D=\D$, we see that $\widehat\mu\in\Cone(\M_V^a)\setminus\{0\}$.

  To prove uniqueness, suppose for a contradiction that there are two minimizers $\mu_1,\mu_2\in\Cone(\M_V^a)$ with $\mu_1\neq\mu_2$ and define $\mu\coloneqq\frac12\mu_1+\frac12\mu_2\in\Cone(\M_V^a)$. Since $V$ is strictly convex, we may apply Lemma \ref{thm:useful_tool} to see that $\VV_\edow(\mu)<\frac12\VV_\edow(\mu_1)+\frac12\VV_\edow(\mu_2)=\min_{\mu\in\D}\VV_\edow(\mu)$, which is the required contradiction. Finally, since $0$ cannot be a minimizer, we must therefore have $\VV_\edow(\widehat\mu)<\VV_\edow(0)=U(\infty)$.\qed
\end{proof}

\begin{remark}\label{rem:one_thing}
  It is worth noting that essentially all of the results in this paper can be proved without appealing to Proposition \ref{thm:lagrange_duality}. Indeed, we can show that $\inf_{y>0}v_\edow(y)$ is attained by a unique $\hat y\in(0,\infty)$ by noting from Corollary \ref{thm:propofv}, later, that $v$ is strictly convex; one can easily rule out the possibility that the infimum is attained at $y=0$ by setting $\mathcal D=\Cone(\M_V^a)$ in Lemma \ref{thm:kabastri}, and at $y=\infty$ by noting that $V'(\infty)=\infty$. Again, from Corollary \ref{thm:propofv} there exists a unique probability measure $\optQQ_{\hat y}$, which attains the infimum in the dual problem \eqref{eqn:simpledual}. From here it is easy to show that $\widehat\mu\coloneqq\hat y\optQQ_{\hat y}$ minimizes $\inf_{\mu\in\D}\VV_\edow(\mu)$. If, in addition, $\M_V^e\neq\emptyset$, then one can show the existence of an optimal trading strategy and that there is no duality gap by applying Proposition \ref{thm:defoptgain} and Theorem \ref{thm:supermart}.
\end{remark}

\begin{proof}[ of Theorem \ref{thm:minimalassumption}(i)]
  We have shown in Proposition \ref{thm:lagrange_duality} that $\M_V^a\neq\emptyset$ implies that $u_\edow<U(\infty)$, so the ``only if'' direction follows easily since $u_\edow^\adm\le u_\edow$. For the case of deterministic endowment, the ``if'' direction follows from Corollary 1 of Biagini and Frittelli (2005). We give here a proof for the case of general endowment. Suppose that $u_\edow^\adm<U(\infty)$. Define $\C^\adm\coloneqq\{X\in L^\infty:X\le(H\cdot S)_T\text{ some }H\in\H^\adm\}$, and $\D^\adm\coloneqq\{\mu\in\ba:\mu(X)\le0\text{ for all }X\in\C^\adm\}$. By applying Theorem 5.6 of Delbaen and Schachermayer (1994), it is easy to show that
  \begin{equation}\label{eqn:repofDadm}
    \D^\adm\cap\ba^\sigma=\begin{cases}\Cone(\M^a),&\text{ if }\M^a\neq\emptyset\\\{0\},&\text{ if }\M^a=\emptyset.\end{cases}
  \end{equation}
  Let
  \[ v_\edow^\adm:=\inf_{\mu\in\Cone(\M^a)}\E{V\left(\radnik[\mu]\right)+\radnik[\mu]\edow}. \]
  Applying the Lagrange Duality Theorem (as in the proof of Proposition \ref{thm:lagrange_duality}), we see that
  \[ \sup_{G\in\C^\adm}\UU_{\edow}(G)=u_\edow^\adm=v_\edow^\adm=\min_{\mu\in\D^\adm}\VV_\edow(\mu). \]

  Since $\min_{\mu\in\D^\adm}\VV_\edow(\mu)=u_\edow^\adm<U(\infty)=V(0)=\VV_\edow(0)$, the minimum on the right hand side cannot be attained by $0$. Lemma \ref{thm:useful_tool} and equation \eqref{eqn:repofDadm} therefore imply that $\M_V^a\neq\emptyset$. An application of Proposition \ref{thm:lagrange_duality} shows that $u_\edow^\adm=v_\edow^\adm=v_\edow=u_\edow$.\qed
\end{proof}

\section{A Dynamic Dual Problem}\label{sec:dyndual}

Throughout this section we fix a stopping time $\tau$, valued in $[0,T]$, and consider a dynamic version of the dual problem on the stochastic interval $\lint\tau,T\rint$. To simplify matters when reading this section for the first time, the reader could consider the case $\tau=0$. Throughout this section we shall assume that Assumptions \ref{ass:rae}, \ref{ass:edow} and \ref{ass:noarb} hold, without further comment. We shall take conditional expectations throughout this section in the general sense of Shiryaev (1995), \S II.7, and say that a random variable $X$ is \emph{conditionally integrable} under a measure $\QQ$, with respect to a sigma algebra $\F$, if $\E[\QQ]{|X|\big|\F}<\infty$, $\QQ$-a.s., so that the conditional expectation $\E[\QQ]{X|\F}$ may be defined.

Given an initial random variable $\xi\in L_+^1(\Omega,\F_\tau,\PP)$, define the simplex of terminal measures
\begin{equation*}
  \ba_\tau(\xi)\coloneqq\{\mu\in\ba_+:\mu(A)=\E{\xi\ind_A}\text{ for all }A\in\F_\tau\}.
\end{equation*}
The convex cone of time-$\tau$ initial random variables for our dynamic dual problem is defined by
\begin{align}
  \CM &\coloneqq \left\{\E{\textstyle\radnik[\mu]\big|\F_\tau}:\mu\in\Cone(\M_V^a)\right\}\notag \\
  &=\left\{\xi\in L_+^1(\Omega,\F_\tau,\PP):\exists\mu\in\Cone(\M_V^a)\text{ s.t.
  }\E{\xi\ind_A}=\mu(A)\:\:\forall A\in\F_\tau\right\}\notag\\
  &= \{\xi\in
  L_+^1(\PP):\ba_\tau(\xi)\cap\Cone(\M_V^a)\neq\emptyset\}.\label{eqn:repofMVtau}
\end{align}
Note that (by definition) if $\xi\in\CM$ then $\xi>0$ $\mu$-a.s. for any $\mu\in\ba_\tau(\xi)$, even if $\PP(\{\xi=0\})>0$. The terminal ``martingale measures'' for the dynamic dual problems are given by
\begin{equation*}
  \M_\tau^a(\xi)\coloneqq\ba_\tau(\xi)\cap\Cone(\M^a)\qquad\text{and}\qquad
  \D_\tau(\xi)\coloneqq\ba_\tau(\xi)\cap\D.
\end{equation*}
These sets represent slices of the cones $\Cone(\M^a)$ and $\D$ respectively.
\begin{definition}[Dual Problems]
  Given $\xi\in\CM$, define
  \begin{alignat}{3}
  v_\tau(\xi)&\coloneqq&\underset{\mu\in\M_\tau^a(\xi)}{\essinf}&\E{V\left(\radnik[\mu]\right)+\radnik[\mu]\edow\bigg|\F_\tau},\label{eqn:dyn_dual}\\
  v(\xi;\tau)&\coloneqq&\quad\inf_{\mu\in\D_\tau(\xi)}&\VV_\edow(\mu).\label{eqn:abstr_dyn_dual}
  \end{alignat}
\end{definition}
For $\edow=0$ and $\xi\in\CM$ strictly positive, the dual problem \eqref{eqn:dyn_dual} coincides with the dynamic dual problem originally defined in Schachermayer (2003). Although our primary aim is to solve the dual problem \eqref{eqn:dyn_dual}, it is more obvious that \eqref{eqn:abstr_dyn_dual} has a solution, due to the lower-semicontinuity of $\VV_\edow$ and following result.
\begin{lemma}\label{thm:compactness}
  The set $\D_\tau(\xi)$ is $\sigma(\ba,L^\infty)$-compact.
\end{lemma}

\begin{proof}
  First note that $\D$ is $\sigma(\ba,L^\infty)$-closed. Indeed, consider a net $(\mu_\alpha)_{\alpha\in A}\subseteq\D$ such that $\mu_\alpha\rightarrow\mu$ in $\sigma(\ba,L^\infty)$. Then for all $G\in \C$ we have $\ip G\mu =\lim_\alpha\ip G{\mu_\alpha}\le0$. We claim that the simplex $\ba_\tau(\xi)$ is $\sigma(\ba,L^\infty)$-compact. To show closedness, consider a net $(\mu_\alpha)_{\alpha\in A}\subseteq\ba_\tau(\xi)$ such that $\mu_\alpha\rightarrow\mu$ in $\sigma(\ba,L^\infty)$. Then for $A\in\F_\tau$ we have $\mu(A)=\lim_\alpha\mu_\alpha(A)=\E{\xi\ind_A}$, so $\mu\in\ba_\tau(\xi)$. Compactness of $\ba_\tau(\xi)$ follows immediately, because it is a subset of the $\sigma(\ba,L^\infty)$-compact ball $\{\mu\in \ba:\|\mu\|\le\|\xi\|_{L^1(\PP)}\}$. It follows that $\D_\tau(\xi)=\ba_\tau(\xi)\cap\D$ is $\sigma(\ba,L^\infty)$-compact.\qed
\end{proof}

\begin{proposition}\label{thm:suboptimal}
  Given any $\xi\in\CM\setminus\{0\}$ there exists a unique probability measure $\widehat\mu_\xi\in\ba_\tau(\xi)\cap\Cone(\M_V^a)$ which attains the infima in the dual problems \eqref{eqn:dyn_dual} and \eqref{eqn:abstr_dyn_dual}. Moreover, $\E{v_\tau(\xi)}=v(\xi;\tau)<\infty$.
\end{proposition}

\begin{proof}
  Since $\xi\in\CM$, we may choose any $\mu$ in the non-empty set $\ba_\tau(\xi)\cap\Cone(\M_V^a)$ by virtue of \eqref{eqn:repofMVtau}. By Lemma \ref{thm:technical} we have $\mu\in\D_\tau(\xi)$. By Assumption \ref{ass:edow} we have that $ \E[\mu]{\edow}\le\|\mu\|x''+\E[\mu]{(H''\cdot S)_T}\le x''<\infty$. It follows from equation \eqref{eqn:abstr_dyn_dual}, Lemma \ref{thm:useful_tool} and the finite relative entropy of $\mu$ that
  \[ v(\xi;\tau)\le\VV_\edow(\mu)=\E{V\left(\radnik[\mu]\right)+\radnik[\mu]\edow}<\infty. \]
  By Lemma \ref{thm:compactness}, $\D_\tau(\xi)$ is $\sigma(\ba,L^\infty)$-compact. Since $\VV_\edow$ is $\sigma(\ba,L^\infty)$-lower semicontinuous, the infimum in \eqref{eqn:abstr_dyn_dual} is attained for some $\widehat\mu_\xi=\widehat\mu(\xi;\tau)\in\D_\tau(\xi)$. Applying Lemma \ref{thm:kabastri}, with $\mathcal D=\D_\tau(\xi)$, we see that $\widehat\mu_\xi\in\Cone(\M_V^a)\setminus\{0\}$.

  We now show that $\widehat\mu_\xi$ is optimal for \eqref{eqn:dyn_dual}. Suppose, for a contradiction, that there exists a $\mu\in\M_\tau^a(\xi)$, and event $A\in\F_\tau$ such that
  \[ \E{\left(V\left(\radnik[\mu]\right)+\radnik[\mu]\edow\right)\ind_A}<\E{\left(V\left(\radnik[\widehat\mu_\xi]\right)+\radnik[\widehat\mu_\xi]\edow\right)\ind_A}. \]
  Define the measure $\widetilde\mu\in\M_\tau^a(\xi)$ by $\widetilde\mu(B)=\mu(A\cap B)+\widehat\mu_\xi(A^c\cap B)$; the proof that $\widetilde\mu\in\Cone(\M^a)$ is left as an exercise for the reader -- this follows from a localization argument, and the fact that any process which is simultaneously a martingale under the probability measures $\widehat\QQ\coloneqq\widehat\mu_\xi/\|\widehat\mu_\xi\|$ and $\QQ\coloneqq\mu/\|\mu\|$ will also be a martingale under $\widetilde\QQ\coloneqq\widetilde\mu/\|\widetilde\mu\|$, noting that $\|\mu\|=\|\widehat\mu_\xi\|=\|\widetilde\mu\|=\E\xi$. Moreover, it is easy to show that $\widetilde\mu$ has finite relative entropy and hence by Lemma \ref{thm:technical} we have $\widetilde\mu\in\D_\tau(\xi)$. Applying Lemma \ref{thm:useful_tool} gives
  \begin{equation*}
    \VV_\edow(\widetilde\mu)=\E{\left(V\left(\radnik[\mu]\right)+\radnik[\mu]\edow\right)\ind_A}
    +\E{\left(V\left(\radnik[\widehat\mu_\xi]\right)+\radnik[\widehat\mu_\xi]\edow\right)\ind_{A^c}}
    <\VV_\edow(\widehat\mu_\xi),
  \end{equation*}
  which contradicts the optimality of $\widehat\mu_\xi$ in the dual problem \eqref{eqn:abstr_dyn_dual}.

  To prove uniqueness in \eqref{eqn:dyn_dual}, suppose for a contradiction that there are two minimizers $\mu_1,\mu_2\in\M_\tau^a(\xi)$ with $\mu_1\neq\mu_2$ and define $\mu\coloneqq\frac12\mu_1+\frac12\mu_2\in\M_\tau^a(\xi)$. Since $V$ is strictly convex,
  \begin{align*}
    &\E{V\left(\radnik[\mu]\right)+\radnik[\mu]\edow\bigg|\F_\tau}\\&\qquad\qquad\qquad
    <\frac12\E{V\left(\radnik[\mu_1]\right)+\radnik[\mu_1]\edow\bigg|\F_\tau}
    +\frac12\E{V\left(\radnik[\mu_2]\right)+\radnik[\mu_2]\edow\bigg|\F_\tau} =v_\tau(y),
  \end{align*}
  the required contradiction. Taking expectations shows that $\widehat\mu_\xi$ is also unique for \eqref{eqn:abstr_dyn_dual}.\qed
\end{proof}

\begin{lemma}\label{thm:convex}
  For $\xi\in\CM\setminus\{0\}$, the map $(0,\infty)\ni y\mapsto v_\tau(y\xi)$ is strictly convex, and the $\F_\tau$-measurable, one-sided derivatives
  \begin{equation*}
    D_-v_\tau(\xi)\coloneqq\lim_{y\nearrow1}\frac{v_\tau(y\xi)-v_\tau(\xi)}{(y-1)\xi}
    \qquad\text{and}\qquad
    D_+v_\tau(\xi)\coloneqq\lim_{y\searrow1}\frac{v_\tau(y\xi)-v_\tau(\xi)}{(y-1)\xi}
  \end{equation*}
  exist $\PP$-a.s. as finite limits on the event $\{\xi>0\}$. Furthermore, given $\xi'\in\CM\setminus\{0\}$ we have $D_+v_\tau(\xi)\le D_-v_\tau(\xi')$ $\PP$-a.s. on the event $\{0<\xi\le\xi'\}$.
\end{lemma}

\begin{proof}
  Suppose that $\xi_0,\xi_1,\tilde\xi\in\CM\setminus\{0\}$, and let $A\coloneqq\{\xi_0<\tilde\xi<\xi_1\}\in\F_\tau$. Let $\widehat\mu_i\in\M_\tau^a(\xi_i)$ for $i=1,2$ be the minimizers in the dynamic dual problems $v_\tau(\xi_0)$, and $v_\tau(\xi_1)$, and pick any $\mu\in\M_\tau^a(\tilde\xi)$. Define the $\F_\tau$-measurable random variable $\lambda\coloneqq\ind_A(\tilde\xi-\xi_0)/(\xi_1-\xi_0)$. Since $\lambda$ is $(0,1)$-valued on the event $A$, we may define the measure $\tilde\mu\in\M_\tau^a(\tilde\xi)$ by
  \[ \tilde\mu(B)=\mu_0((1-\lambda)\ind_{A\cap B})+\mu_1(\lambda\ind_{A\cap B})+\mu(A^c\cap B). \]
  Furthermore,
  \begin{align}\label{eqn:convexx}
    &((1-\lambda)v_\tau(\xi_0)+\lambda v_\tau(\xi_1))\ind_A\notag\\
    &\qquad=(1-\lambda)\ind_A\E{V\left(\radnik[\widehat\mu_0]\right)+\radnik[\widehat\mu_0]\edow\bigg|\F_\tau}
    +\lambda\ind_A\E{V\left(\radnik[\widehat\mu_1]\right)+\radnik[\widehat\mu_1]\edow\bigg|\F_\tau}\notag\\
    &\qquad>\E{V\left((1-\lambda)\radnik[\widehat\mu_0]
    +\lambda\radnik[\widehat\mu_1]\right)\ind_A+\left((1-\lambda)\radnik[\widehat\mu_0]
    +\lambda\radnik[\widehat\mu_1]\right)\edow\ind_A\bigg|\F_\tau}\notag\\
    &\qquad=\E{\left(V\left(\radnik[\tilde\mu]\right)+\radnik[\tilde\mu]\edow\right)\ind_A\bigg|\F_\tau}
    \ge v_\tau(\tilde\xi)\ind_A.
  \end{align}
  Strict convexity of the maps $\RR\ni\lambda\mapsto v_\tau(\lambda\xi)$ follows from \eqref{eqn:convexx} by choosing $\xi_0,\xi_1$ and $\tilde\xi$ to be scalar multiples of $\xi$. It is thus clear that the one-sided derivatives $D_-v_\tau(\xi)$ and $D_+v_\tau(\xi)$ exist and are finite on the event $\{\xi>0\}$ because the ratios used in their definitions are finite and increasing $\PP$-a.s. with respect to $y$.

  Rearranging \eqref{eqn:convexx} gives that
  \begin{equation}\label{eqn:convexxx}
    \frac{v_\tau(\tilde\xi)-v_\tau(\xi_0)}{\tilde\xi-\xi_0}\ind_A
    <\frac{v_\tau(\xi_1)-v_\tau(\xi_0)}{\xi_1-\xi_0}\ind_A.
  \end{equation}
  It follows, by applying \eqref{eqn:convexxx} twice on the event $A=\{0<\xi\le\xi'\}$, that
  \begin{equation*}
    D_-v(\xi)\ind_A=\lim_{y\nearrow1}\frac{v_\tau(\xi)-v_\tau(y\xi)}{(1-y)\xi}\ind_A \le \lim_{y\searrow1}\frac{v_\tau(y\xi')-v_\tau(\xi')}{(y-1)\xi'}\ind_A=D_+v(\xi')\ind_A.
  \end{equation*}
  \qed
\end{proof}

\begin{remark}
  Note that the map $\xi\mapsto v_\tau(\xi)$ is not randomly convex in the sense of Schachermayer (2003) or Biagini and Frittelli (2007). This difference results from the fact that the domain of our dual problem \eqref{eqn:dyn_dual} consists exclusively of densities of local-martingale measures for the whole time interval $[0,T]$ -- a set which is typically not closed under random convex combinations.
\end{remark}

If the one sided derivatives are equal then we say that the two-sided derivative $Dv_\tau(\xi)$ exists.

\begin{proposition}\label{thm:newa17}
  For $\xi\in\CM\setminus\{0\}$, the two sided derivative $Dv_\tau(\xi)$ exists and
  \begin{equation}\label{eqn:Dv}
    Dv_\tau(\xi)\le\E[\QQ]{V'\left(\radnik[\widehat\mu_\xi]\right)+\edow\bigg|\F_\tau},
  \end{equation}
  $\PP$-a.s. on the event $\{\xi>0\}$ for all $\QQ\in\M_\tau^a(\xi/\|\xi\|_1)$ such that $\E{V\left(\radnik\right)+\radnik\edow\big|\F_\tau}<\infty$ $\PP$-a.s. Setting $\QQ=\optQQ_\xi=\widehat\mu_\xi/\|\widehat\mu_\xi\|$ gives equality in \eqref{eqn:Dv}.
\end{proposition}

\begin{remark}
  The observant reader may be concerned that the conditional expectation on the right hand side of equation \eqref{eqn:Dv} is unique only up to $\QQ$-null sets. This does not cause any problems however because $\E{\radnik|\F_\tau}=\xi/\|\xi\|_1$, so $\PP$ and $\QQ$ are equivalent on the sigma algebra $\{A\cap\{\xi>0\}:A\in\F_\tau\}$.
\end{remark}

\begin{proof}
  For $\lambda>0$ define $\mu_\lambda\coloneqq\lambda\widehat\mu_\xi\in\M_\tau^a(\lambda\xi)$, $\nu_\lambda\coloneqq V\left(\radnik[\mu_\lambda]\right)+\radnik[\mu_\lambda]\edow$. Define $h(\lambda)\coloneqq\E{\nu_\lambda\big|\F_\tau}$. Since $\lambda\mapsto\nu_\lambda$ is convex, the integrable random variables $(\nu_\lambda-\nu_1)/(\lambda-1)$ are monotone increasing with respect to the parameter $\lambda\neq1$. Therefore, by the monotone convergence theorem we may differentiate $h(\lambda)$ at $\lambda=1$ inside the expectation. Moreover, since $h(1)=v_\tau(\xi)$, $h(\lambda)\ge v_\tau(\lambda\xi)$ for all $\lambda>0$ and, from Lemma \ref{thm:convex}, the map $\lambda\mapsto v_\tau(\lambda\xi)$ is convex, it follows that the two-sided derivative $D v_\tau(\xi)$ exists and
  \begin{equation*}
    \xi D v_\tau(\xi)=h'(1)
    =\E{\radnik[\widehat\mu_\xi]V'\left(\radnik[\widehat\mu_\xi]\right)+\radnik[\widehat\mu_\xi]\edow\bigg|\F_\tau}
    =\xi\E[\widehat\QQ_\xi]{V'\left(\radnik[\widehat\mu_\xi]\right)+\edow\bigg|\F_\tau},
  \end{equation*}
  where we have used Bayes' rule in the last equality. Note that the conditional expectations on the right hand side are well defined, as the random variable in the first expectation is dominated above and below by the integrable random variables $\nu_2-\nu_1$ and $2(\nu_1-\nu_{1/2})$ respectively.

  Take any $\QQ\in\M_\tau^a(\xi/\|\xi\|_1)$ satisfying $\E{V\left(\radnik\right)+\radnik\edow\big|\F_\tau}<\infty$ a.s. Redefine, for $\lambda\ge1$,
  $\mu_\lambda\coloneqq\widehat\mu_\xi+(\lambda-1)\|\xi\|_1\QQ\in\M_\tau^a(\lambda\xi)$, which also satisfies $\E{\nu_\lambda\big|\F_\tau} = \E{V\left(\radnik[\mu_\lambda]\right)+\radnik[\mu_\lambda]\edow\big|\F_\tau} <\infty$ a.s. (that is, $\nu_\tau$ is conditionally integrable with respect to the sigma algebra $\F_\tau$). By repeating the convexity argument above, it follows from a conditional version of the monotone convergence theorem that we may evaluate the right-sided derivative, $D_+h$, of $h$ at $\lambda=1$ by differentiating inside the expectation. This time, since $h(1)=v_\tau(\xi)$, $h(\lambda)\ge v_\tau(\lambda\xi)$ for all $\lambda\ge1$ and the map $\lambda\mapsto v_\tau(\xi)$ is convex, it follows that
  \begin{equation*}
    \xi D_+v_\tau(\xi)\le D_+h(1)
    =\E{\radnik[\mu]V'\left(\radnik[\widehat\mu_\xi]\right)+\radnik[\mu]\edow\bigg|\F_\tau}
    =\xi\E[\QQ]{V'\left(\radnik[\widehat\mu_\xi]\right)+\edow\bigg|\F_\tau}.
  \end{equation*}
  The conditional expectations on the right hand side are well defined, as the random variable in the first expectation is dominated above by the conditionally integrable random variable $\nu_2-\nu_1$. Since $D_+v_\tau(\xi)=Dv_\tau(\xi)$ is real-valued, the random variable $V'\big(\radnik[\widehat\mu_\xi]\big)$ is in fact conditionally integrable under $\QQ$.\qed
\end{proof}

\begin{corollary}[Growth condition on $v_\tau$]\label{thm:2ndingredient}
  There exists a constant $C'>0$ such that for any stopping time $\tau$, valued in $[0,T]$, $\xi D v_\tau(\xi)\le C'(v_\tau(\xi)+\xi)$, $\PP$-a.s. for all $\xi\in\CM\setminus\{0\}$.
\end{corollary}

\begin{proof}
  Let $\widehat\mu_\xi=\widehat\mu(\xi;\tau)$ denote the minimizer in the dynamic dual problem. Applying Proposition \ref{thm:newa17}, Assumption \ref{ass:growth}(ii) and Assumption \ref{ass:edow}, we may find a $C>1$ such that $\PP$-a.s.
  \begin{align*}
    \xi Dv_\tau(\xi) &= \E{\radnik[\widehat\mu_\xi]V'\left(\radnik[\widehat\mu_\xi]\right)+\radnik[\widehat\mu_\xi]\edow\bigg|\F_\tau}\\
    &\le\E{CV\left(\radnik[\widehat\mu_\xi]\right)+\radnik[\widehat\mu_\xi]\edow\bigg|\F_\tau}\\
    &= Cv_\tau(\xi)-(C-1)\E{\radnik[\widehat\mu_\xi]\edow\bigg|\F_\tau}\\
    &\le Cv_\tau(\xi)-(C-1)x'\xi.
  \end{align*}
  \qed
\end{proof}

\begin{corollary}\label{thm:propofv}
  For $y>0$ there exists a unique probability measure $\optQQ_y\in\M_V^a$ which attains the finite infimum in the dual problem \eqref{eqn:simpledual}. The dual function $v_\edow(y)$ is continuously differentiable and strictly convex as a function of $y>0$. Moreover,
  \[ v'_\edow(y)=\min_{\QQ\in\M_V^a}\E[\QQ]{V'\left(y\radnik[\optQQ_y]\right)+\edow}, \]
  in which $\optQQ_y$ is also optimal. Finally, $v_\edow(y)$ satisfies the growth condition $yv'_\edow(y)\le C'(v_\edow(y)+y)$.
\end{corollary}

\begin{proof}
  The properties of $v_\edow(y)$ follow from the above results by setting $\tau=0$, $\xi=y$ and defining $\optQQ_y=\widehat\mu_y/y$. If $\F_0$ consists of more than just $\PP$-null sets, one can replace it by the trivial $\sigma$-algebra $\{\emptyset,\Omega\}$, and the results still hold. In particular, the continuous differentiability of $v_\edow(y)$ follows from convexity and differentiability.\qed
\end{proof}

\section{The Optimal Trading Strategy}\label{sec:opttrading}

In this section, we demonstrate the existence of an optimal trading strategy in the primal problem. We assume throughout this section that Assumptions \ref{ass:rae}, \ref{ass:edow} and \ref{ass:noarb} hold. Throughout this section we let $\widehat\mu\in\Cone(\M_V^a)\setminus\{0\}$ be the unique minimizer in the dual problem \eqref{eqn:maindual}, and let $\optQQ$ denote the normalization of $\widehat\mu$. Throughout this section, we consider all stochastic integrals, and almost-sure statements to be taken within the filtered probability space $(\Omega,\F,(\F_t)_{t\in[0,T]},\optQQ)$.
\begin{proposition}\label{thm:defoptgain}
  There exists a predictable process $\widehat H\in L(S;\optQQ)$ such that $\widehat H\cdot S$ is a $\optQQ$-martingale and such that, almost surely,
  \[ U'\left((\widehat H\cdot S)_T+\edow\right)=\radnik[\widehat\mu]\qquad\text{or equivalently} \qquad (\widehat H\cdot S)_T=-V'\left(\radnik[\widehat\mu]\right)-\edow. \]
  Moreover,
  \[ (\widehat H\cdot S)_\tau=-Dv_\tau\left(\E{\radnik[\widehat\mu]\bigg|\F_\tau}\right). \]
\end{proposition}

\begin{proof}
  Define $\widehat G\coloneqq-V'\big(\radnik[\widehat\mu]\big)-\edow$. From Corollary \ref{thm:propofv}, we have $\widehat G\in L^1(\optQQ)$ and
  \begin{equation*}
    \bigE[\optQQ]{\widehat G}=-\E[\optQQ]{V'\left(\radnik[\widehat\mu]\right)+\edow}=-v'_\edow(\|\widehat\mu\|)=0,
  \end{equation*}
  where the last equality follows from Proposition \ref{thm:lagrange_duality} and differentiability of $v_\edow(y)$. Similarly, for any $\QQ\in\M_V^a$ we have $\widehat G\in L^1(\QQ)$ and $\bigE[\QQ]{\widehat G}\le0$.

  Applying Theorem 4 of Biagini and Frittelli (2005), we get $\widehat G\in\normcl[\optQQ]{\Kperm-L_+^1(\optQQ)}$, where $\Kperm\coloneqq\{(H\cdot S)_T:H\in\H^\perm\}$. Hence there exist sequences $G^n\in\Kperm$, $F^n\in L_+^1(\optQQ)$, and $H^n\in\H^\perm$ such that $(H^n\cdot S)_T=G^n$ and $G^n-F^n\rightarrow\widehat G$ in $L^1(\optQQ)$ as $n\rightarrow\infty$. Since $\E[\optQQ]{G^n}\le0$ and $\bigE[\optQQ]{\widehat G}=0$ we see that $\E[\optQQ]{|F^n|}\le\bigE[\optQQ]{F^n-G^n+\widehat G} \le\bigE[\optQQ]{|\widehat G-(G^n-F^n)|} \rightarrow0$ as $n\rightarrow\infty$. Hence $\bigE[\optQQ]{|\widehat G-G^n|} \le\bigE[\optQQ]{|\widehat G-(G^n-F^n)|}+\E[\optQQ]{|F^n|}\rightarrow0$ as $n\rightarrow\infty$, so $\widehat G\in\normcl[\optQQ]{\Kperm}$. By passing to a subsequence if necessary, we may assume that $G^n\rightarrow\widehat G$ as $n\rightarrow\infty$, a.s., and that $\sum_n\|G^n-G^{n+1}\|_{L^1(\optQQ)}<\infty$. We may also assume that $G^1=-1$. Let $w\in L^1(\optQQ)$ be defined by $ w\coloneqq-\inf_n G^n\le1+\sum_n|G^n-G^{n+1}|$.

  Using the supermartingale property of the processes $H^n\cdot S$, we have $(H^n\cdot S)_t\ge\E[\optQQ]{(H^n\cdot S)_T|\F_t}\ge -\E[\optQQ]{w|\F_t}$. We may apply Theorem 14.5.13 of Delbaen and Schachermayer (2006) (c.f. Delbaen and Schachermayer (1998, 1999)): There exist convex combinations $K^n\in\operatorname{conv}\{H^n,H^{n+1},\dots\}$, a $\optQQ$-supermartingale $(V_t)_{t\in[0,T]}$ with $V_0\le0$ and an $\widehat H\in L(S;\optQQ)$ such that $\widehat H\cdot S$ is a local $\optQQ$-martingale and a $\optQQ$-supermartingale, $V_T=\underset{n\rightarrow\infty}\lim(K^n\cdot S)_T$ a.s. and $(\widehat H\cdot S)_t\ge V_t$ for all $t$, a.s. Since $(K^n\cdot S)_T\in\operatorname{conv}\{G^n,G^{n+1},\dots\}$, it follows that $V_T=\widehat G$. Hence, $ 0\ge\bigE[\optQQ]{(\widehat H\cdot S)_T}\ge\bigE[\optQQ]{V_T} =\bigE[\optQQ]{\widehat G}=0,$ so $\widehat G=(\widehat H\cdot S)_T$ a.s. and the process $\widehat H\cdot S$ is a $\optQQ$-martingale.

  It now follows that
  \[ (\widehat H\cdot S)_\tau = \E[\optQQ]{(\widehat H\cdot S)_T\big|\F_\tau} = -\E[\optQQ]{V'\left(\radnik[\widehat\mu]\right)+\edow\bigg|\F_\tau} = -Dv_\tau\left(\E{\radnik[\widehat\mu]\bigg|\F_\tau}\right), \]
  where the last equality is true because $\widehat\mu$ is optimal for \eqref{eqn:dyn_dual} with $\xi=\bigE{\radnik[\widehat\mu]\big|\F_\tau}$.\qed
\end{proof}

\begin{remark}\label{rem:aboutequiv}
  Note that by applying Lemma \ref{thm:kabastri} with $\mathcal D=\D$ we see that $\QQ\ll\optQQ$ for all $\QQ\in\M_V^a$. It follows from Theorem IV.25 of Protter (2003) that $\widehat H\in L(S;\QQ)$ for all $\QQ\in\M_V^a$. In the case when $\M_V^e\neq\emptyset$ we also have that $\widehat\QQ\sim\PP$, by Lemma \ref{thm:kabastri}, so $\widehat H\in L(S;\PP)$.
\end{remark}

We now turn to the problem of showing that the wealth process of $\widehat H$ is a supermartingale under all local martingale measures with finite relative entropy. Theorem \ref{thm:supermart} deals with this result in generality, but relies on some powerful results concerning a characterization of supermartingales, taken from Schachermayer (2003). For the case of exponential utility however, a marginally weaker version of the supermartingale property can be proved by elementary considerations. Although separate attention is given to the case of exponential utility in Schachermayer (2003), our proof seems to be new, even for the case when $\edow$ is deterministic.

The proofs below use our analysis of the dynamic dual problem of Section \ref{sec:dyndual}. In the next result we consider the case when $U(x)=-\frac1\gamma\exp(-\gamma x)$. The constant $\gamma>0$ represents the investor's absolute risk aversion. Note that for our proof we only need to apply the results of Section \ref{sec:dyndual} for deterministic stopping times $\tau=t$.
\begin{proposition}[Exponential Utility]\label{thm:exponly}
  Suppose that $\M_V^e\neq\emptyset$. Then
  \begin{equation}\label{eqn:toprove}
    (\widehat H\cdot S)_t = \underset{\QQ\in\M_V^e}\essmax\E[\QQ]{\frac1\gamma\ln\left(\frac{\d\PP}{\d\widehat\mu}\right)-\edow\bigg|\F_t}
  \end{equation}
  and hence the optimal wealth process, $\widehat H\cdot S$, is a $\QQ$-supermartingale for all $\QQ\in\M_V^e$.
\end{proposition}

\begin{proof}
  For the special case of exponential utility we have $V(y)=\frac y\gamma(\ln y-1)$. First
  we introduce some notation: Given two measures $\QQ_1\ll\PP$ and $\QQ_2\sim\PP$, we define
  the concatenated measure $\QQ_1\otimes_t\QQ_2\ll\PP$ via its Radon-Nikod\'ym derivative
  \[ \radnik[(\QQ_1\otimes_t\QQ_2)] \coloneqq \E{\radnik[\QQ_1]\bigg|\F_t}\frac{\radnik[\QQ_2]}{\E{\radnik[\QQ_2]\big|\F_t}}. \]
  Define the probability measure $\optQQ\coloneqq\widehat\mu/\|\widehat\mu\|$. Given any $\widetilde\QQ\in\M_V^e$ we have $\optQQ\otimes_t\widetilde\QQ\in\M^e$ and $\radnik[(\optQQ\otimes_t\widetilde\QQ)]=YZ$ where $Y\coloneqq\E{\radnik[\optQQ]|\F_t}\big/\E{\radnik[\widetilde\QQ]|\F_t}$ is $\F_t$-measurable, and $Z=\radnik[\widetilde\QQ]$. Now $\E{YZ\ln(YZ)|\F_t}=Y\ln Y\E{Z|\F_t}+Y\E{Z\ln Z|\F_t}<\infty$ a.s. and, by Assumption \ref{ass:edow}, $\E{\radnik[(\optQQ\otimes_t\widetilde\QQ)]\edow}\le x''+\E[\optQQ\otimes_t\widetilde\QQ]{(H''\cdot S)_T} \le x''+\E[\optQQ]{(H''\cdot S)_t}=x''<\infty$. Since $\widehat\mu$ is also the minimizer in \eqref{eqn:dyn_dual} and \eqref{eqn:abstr_dyn_dual} with $\tau=t$ and $\xi=\bigE{\radnik[\widehat\mu]\big|\F_t}$ it follows from Proposition \ref{thm:newa17} that
  \begin{align*}
    Dv_t\left(\E{\radnik[\widehat\mu]\bigg|\F_t}\right)\le \E[\optQQ\otimes_t\widetilde\QQ]{\frac1\gamma\ln\left(\radnik[\widehat\mu]\right)+\edow\bigg|\F_t}
    =\E[\widetilde\QQ]{\frac1\gamma\ln\left(\radnik[\widehat\mu]\right)+\edow\bigg|\F_t}.
  \end{align*}
  Applying Proposition \ref{thm:newa17}, we have equality for the optimal measure $\optQQ$. Equation \eqref{eqn:toprove} follows by applying Proposition \ref{thm:defoptgain}. Since $\widehat H\cdot S$ is a $\optQQ$-martingale, we have
  \begin{align*}
    \underset{\QQ\in\M_V^e}\esssup\E[\QQ]{(\widehat H\cdot S)_t\big|\F_s}
    &=\underset{\QQ\in\M_V^e}\esssup\E[\QQ]{\E[\widehat\QQ]{(\widehat H\cdot S)_T\big|\F_t}\big|\F_s}\\
    &=\underset{\QQ\in\M_V^e}\esssup\E[\QQ\otimes_t\widehat\QQ]{(\widehat H\cdot S)_T\big|\F_s}\\
    &\le\underset{\QQ\in\M_V^e}\esssup\E[\QQ]{(\widehat H\cdot S)_T\big|\F_s}=(\widehat H\cdot S)_s.
  \end{align*}
  The inequality above follows because $\QQ\otimes_t\widehat\QQ\in\M_V^e$ by Lemma 2 of Schachermayer (2003).\qed
\end{proof}

We now present the general case, whose proof follows the remarkable ideas presented in Schachermayer (2003):
\begin{theorem}[General Utility]\label{thm:supermart}
  The optimal wealth process, $\widehat H\cdot S$, is a $\QQ$-su\-per\-mar\-tin\-gale for all $\QQ\in\M_V^a$. If $\M_V^e\neq\emptyset$ then $\widehat H\in\H^\perm$.
\end{theorem}

\begin{proof}
  Suppose that there exists a $\QQ\in\M^a$ such that $\widehat H\cdot S$ fails to be a $\QQ$-supermartingale. Then applying Lemma 1 of Schachermayer (2003), it is possible to find a sequence of $[0,T]\cup\{\infty\}$-valued stopping times increasing to $+\infty$ such that $(\widehat H\cdot S)_{\tau_n}\le0$ on $\{\tau_n<\infty\}$ and
  \begin{equation}\label{eqn:superm}
    \lim_{n\rightarrow\infty}\E[\QQ]{(\widehat H\cdot S)_{\tau_n}\ind_{\{\tau_n<\infty\}}}<0.
  \end{equation}
  Since $H\cdot S$ is a uniformly integrable martingale under $\optQQ$ we also have
  \begin{equation}\label{eqn:martingale}
    \lim_{n\rightarrow\infty}\E[\optQQ]{(\widehat H\cdot S)_{\tau_n}\ind_{\{\tau_n<\infty\}}}=0.
  \end{equation}
  It follows from \eqref{eqn:superm} and \eqref{eqn:martingale} that
  \[ \limsup_{n\rightarrow\infty}\E{(\widehat H\cdot S)_{\tau_n}Y_{\tau_n}\ind_{\{\tau_n<\infty\}}\ind_{\{Y_{\tau_n} \ge \widehat Y_{\tau_n}\}}}<0, \]
  where $Y_t$ is the density process of $y\radnik$, and $\widehat Y_t$ is the density process of $\radnik[\widehat\mu]$. Set $A_n\coloneqq\{\tau_n<\infty,Y_{\tau_n}\ge\widehat Y_{\tau_n}\}$, and note that $\limsup_{n\rightarrow\infty}A_n=\emptyset$. By applying Lemma \ref{thm:kabastri}, we see that $y\QQ\ll\widehat\mu$, and hence $\{\widehat Y_{\tau_n}=0\}\subseteq\{Y_{\tau_n}=0\}$. It follows from the monotonicity of $Dv_{\tau_n}$ (shown in Lemma \ref{thm:convex}) that $Y_{\tau_n}Dv_{\tau_n}(\widehat Y_{\tau_n})\le Y_{\tau_n}Dv_{\tau_n}(Y_{\tau_n})$ $\PP$-a.s. Applying Corollary \ref{thm:2ndingredient} and Proposition \ref{thm:defoptgain}
  \begin{align*}
    \E{\left(V\left(y\radnik\right)+y\radnik\edow\right)\ind_{A_n}}
    &=\E{\E{V\left(y\radnik\right)+y\radnik\edow\bigg|\F_{\tau_n}}\ind_{A_n}}\\
    &\ge\E{v_{\tau_n}(Y_{\tau_n})\ind_{A_n}}\\
    &\ge c\E{Y_{\tau_n}Dv_{\tau_n}(Y_{\tau_n})\ind_{A_n}}-\E{Y_{\tau_n}\ind_{A_n}}\\
    &\ge c\E{Y_{\tau_n}Dv_{\tau_n}(\widehat Y_{\tau_n})\ind_{A_n}}-y\QQ(A_n)\\
    &\ge cy\E[\QQ]{Dv_{\tau_n}(\widehat Y_{\tau_n})\ind_{A_n}}-y\QQ(A_n)\\
    &=-cy\E[\QQ]{(\widehat H\cdot S)_{\tau_n}\ind_{A_n}}-y\QQ(A_n).
  \end{align*}
  Sending $n\rightarrow\infty$, we deduce from \eqref{eqn:superm} that $\QQ$ cannot have finite relative entropy.\qed
\end{proof}

\section{Dependence on the Endowment}\label{sec:dependow}

In the next result, we investigate the properties of $u_\edow$ as a function of the endowment. This study is worthwhile in its own right, but will also help us prove Proposition \ref{thm:indifprice}. Assumptions \ref{ass:rae} and \ref{ass:noarb} hold throughout this section. We will also assume throughout that all endowments satisfy Assumption \ref{ass:newweakerendowment}. Let $\widehat\mu_\edow$ denote the optimal measure in the dual problem with endowment $\edow$, and let $\optQQ_\edow$ denote the normalization of $\widehat\mu_\edow$.
\begin{proposition}\label{thm:deponendow}
  Define $u(\edow)\coloneqq u_{\edow}$. Then
  \begin{enumerate}
    \item (Monotonicity) $u(\edow_1)\le u(\edow_2)$ for any endowments $\edow_1\le \edow_2$.
    Moreover, $u(\edow_1)<u(\edow_1+\epsilon)$ for all $\epsilon>0$. If, in addition,
    $\M_V^e\neq\emptyset$ then $u(\edow_1)<u(\edow_2)$ for any endowments $\edow_1\le \edow_2$ such
    that $\edow_1\neq \edow_2$;

    \item (Concavity) Given endowments $\edow_1,\edow_2$ and $\lambda\in[0,1]$,
    \[ u(\lambda \edow_1+(1-\lambda)\edow_2)\ge\lambda u(\edow_1)+(1-\lambda)u(\edow_2); \]

    \item (Strong Continuity) If $(\edow_n)_{n\in\NN}$ is a sequence of endowments such that
    \begin{equation}\label{eqn:surprise}
      \sup_{\QQ\in\M_V^a}\E[\QQ]{\edow_n-\edow}\rightarrow0\text{ and }
      \inf_{\QQ\in\M_V^a}\E[\QQ]{\edow_n-\edow}\rightarrow0
    \end{equation}
    as $n\rightarrow\infty$ then $u(\edow_n)\rightarrow u(\edow)$ as $n\rightarrow\infty$.

    Moreover, $\widehat\mu_{\edow_n}\rightarrow\widehat\mu_\edow$ weakly in total variation (i.e. in $\sigma(\ba^\sigma,L^\infty)$) as $n\rightarrow\infty$.

    \item (Lebesgue Continuity) Suppose that $(\edow_n)_{n\ge0}$ is a sequence of endowments which uniformly satisfy Assumption \ref{ass:newweakerendowment} (in the sense that the sub- and super-replicating strategies $(x',H')$ and $(x'',H'')$ can be chosen independently of $n$). If $\edow_n\rightarrow \edow$, $\PP$-a.s as $n\rightarrow\infty$ then $u(\edow_n)\rightarrow u(\edow)$ as $n\rightarrow\infty$.
  \end{enumerate}
\end{proposition}
The proof of Proposition \ref{thm:deponendow} is given in the Appendix.

\begin{remark}
  The Lebesgue Continuity property above is equivalent to Fatou properties from both above and below (see \eqref{eqn:fatou1} and \eqref{eqn:fatou2}), or, alternatively, to continuity from both
  above and below in the sense that
  \begin{gather*}
    \edow_n\searrow \edow\quad \PP\text{-a.s.}\qquad\Longrightarrow\qquad u(\edow_n)\searrow u(\edow)\\
    \intertext{and} \edow_n\nearrow \edow\quad \PP\text{-a.s.}\qquad\Longrightarrow\qquad u(\edow_n)\nearrow u(\edow).
  \end{gather*}
  The terminology ``Lebesgue Continuity'' has been motivated by similar terminology for utility functions defined on $L^\infty$, which can be found in Jouini, Schachermayer and Touzi (2006).
\end{remark}

\begin{corollary}\label{thm:one}
  Define $u_\edow(B)=u_{\edow+B}$. Then $u_\edow(B)$ satisfies the properties (i)-(iv) of Proposition \ref{thm:deponendow}. Furthermore,
  \[ u_\edow\big(B-\E[\optQQ_\edow]B\big)\le u_\edow\le  u_\edow\big(B-\inf_{\QQ\in\M_V^a}\E[\QQ]B\big). \]
\end{corollary}

\begin{proof}
  Since $u_\edow(B)=u(\edow+B)$, it is clear that conditions (i)-(iv) of Proposition \ref{thm:deponendow} are satisfied. Furthermore, by Theorem \ref{thm:main},
  \begin{align*}
    u_\edow\big(B-\inf_{\QQ\in\M_V^a}\E[\QQ]B\big)
    &=\min_{\mu\in\Cone(\M_V^a)}\left\{\E{V\left(\radnik[\mu]\right)+\radnik[\mu]\edow}+\E{\radnik[\mu]\big(B-\inf_{\QQ\in\M_V^a}\E[\QQ]B\big)}\right\}\\
    &\ge v_\edow+0=u_\edow.
  \end{align*}
  For the other inequality,
  \begin{align*}
    u_\edow(B-\E[\optQQ_\edow]B) &=
    \min_{\mu\in\Cone(\M_V^a)}\left\{\E{V\left(\radnik[\mu]\right)+\radnik[\mu]\big(\edow+B-\E[\optQQ_\edow]B\big)}\right\}\\
    &\le
    \E{V\left(\radnik[\widehat\mu_\edow]\right)+\radnik[\widehat\mu_\edow]\edow}+\E{\radnik[\widehat\mu_\edow]\big(B-\E[\optQQ_\edow]B\big)}\\
    &= v_\edow+0=u_\edow.
  \end{align*}
  \qed
\end{proof}

\begin{remark}
  Note that for $x\in\RR$, the function $u_\edow(x)$ itself inherits the properties of a utility function. Furthermore,
  \begin{align*}
    u_\edow(x)=v_{\edow+x}&=\min_{\mu\in\Cone(\M_V^a)\setminus\{0\}}\E{V\left(\radnik[\mu]\right)+\radnik[\mu](\edow+x)}\\
    &=\min_{y>0}\left\{\inf_{\QQ\in\M_V^a}\E{V\left(y\radnik\right)+y\radnik\edow}+xy\right\}\\
  &=\min_{y>0}\{v_\edow(y)+xy\},
  \end{align*}
  where $v_\edow(y)$ is defined in equation \eqref{eqn:simpledual}, so $v_\edow(y)$ is the convex conjugate of $u_\edow(x)$. Moreover $v_\edow(y)$ satisfies the following growth condition: For any $\lambda>0$,
  \begin{align*}
    v_\edow(\lambda y) &= \inf_{\QQ\in\M_V^a}\E{V\left(\lambda y\radnik\right)+\lambda y\radnik\edow} \\
    &\le \inf_{\QQ\in\M_V^a}\E{CV\left(y\radnik\right)+\lambda y\radnik\edow}\\
    &\le C\inf_{\QQ\in\M_V^a}\E{V\left(y\radnik\right)+y\radnik\edow}-\inf_{\QQ\in\M_V^a}\E[\QQ]{(C-\lambda)\edow}y\\
    &=Cv_\edow(y)+C'y.
  \end{align*}
  It follows that $u_\edow(x)$ satisfies the condition of Reasonable Asymptotic Elasticity, i.e. Assumption \ref{ass:rae}.
\end{remark}

\section{Marginal Utility-Based Price Processes}\label{sec:mubpp}

In this section we consider the dynamic pricing of purely financial assets. Each of the financial assets, to be introduced to an existing economy, is assumed to be in zero total supply. The dynamics of the asset prices should be consistent with the balance of supply and demand, expected to hold for a financial market in equilibrium. The price processes of the new assets should therefore be such that a representative investor acting optimally will not wish to invest in them.

The precise definition of a MUBPP is given in Definition \ref{def:MUBPP}. The intuition behind a MUBPP can alternatively be viewed at the level of an individual investor, whose expected utility would not increase if the assets are introduced at a fair price. Put differently, if an individual investor believes the asset price process is fair then their optimal demand for the new assets will be equal to zero. Our main result in this section is that a locally bounded semimartingale is a MUBPP if and only if it is a local martingale under the (normalized) optimal measure in the dual problem.

Marginal Utility-Based Prices have previously been investigated by Hugonnier, Kramkov and Schachermayer (2005), in the context of static buy-and-hold strategies for the new assets, and deterministic initial wealth. A dynamic version of this theory was developed by Kallsen (2002) and Kallsen and K\"uhn (2004, 2005), where they refer to ``Neutral Price Processes''. In the most recent of these papers they show, in the discrete time setting, that Neutral Price Processes are martingales under the optimal dual measure. We refer the reader to Foldes (2000) for a wider exposition on themes within the economic theory of asset pricing. It is worth noting that there are also close links between MUBPP's and Davis's ``fair price'' (see Davis (1997) and Karatzas and Kou (1996)).

Let $S$, $S'$ be respectively $d$- and $d'$-dimensional locally bounded semimartingales. We can combine $S$ and $S'$ together to form the $(d+d')$-dimensional locally bounded semimartingale $(S,S')$, which represents an augmented market. Note that $\M_V^a(S,S')=\M_V^a(S)\cap\M_V^a(S')$. We shall assume throughout this section that $u_\edow^\adm(S)<U(\infty)$, which by Theorem \ref{thm:minimalassumption} is equivalent to the statement that $\M_V^a(S)$ is non-empty.
\begin{lemma}
  Consider an investor with a utility function $U$ satisfying Assumption \ref{ass:rae} and a random endowment $\edow$ satisfying Assumption \ref{ass:edow}. Assume that the investor has access to a market $S$. Let $S'$ be a $d'$-dimensional locally bounded semimartingale. The following statements are equivalent:
  \begin{enumerate}
    \item $u_{\edow}(S)=u_{\edow}^\adm(S,S')$;
    \item $\M_V^a(S,S')\neq\emptyset$ and $u_{\edow}(S)=u_{\edow}(S,S')$;
  \end{enumerate}
  If, furthermore, $\M_V^e(S)$ is non-empty then the above statements are equivalent to
  \begin{enumerate}\setcounter{enumi}2
    \item $\M_V^e(S,S')\neq\emptyset$ and $(\widehat H,0)$ is optimal in the primal problem $u_{\edow}(S,S')$, where $\widehat H\in\H^\perm$ is the optimal solution to the primal problem $u_{\edow}(S)$.
  \end{enumerate}
\end{lemma}

\begin{proof}
  The equivalence of conditions (i) and (ii) is trivial by Theorem \ref{thm:minimalassumption}(i). Suppose that $\M_V^e(S)$ is non-empty, in which case, by Theorem \ref{thm:main} there exists an optimal $\widehat H\in\H^\perm$ in the primal problem $u_\edow(S)$. The implication (iii)$\Rightarrow$(ii) is trivial, so we shall only demonstrate (ii)$\Rightarrow$(iii): $(\widehat H,0)$ is optimal in the primal problem $u_\edow(S,S')$ because
  \[ u_\edow(S,S') = u_\edow(S) =\E{U\big((\widehat H\cdot S)_T+\edow\big)} = \E{U\big(((\widehat H,0)\cdot(S,S'))_T+\edow\big)}. \]
  It now follows from Theorem \ref{thm:minimalassumption}(ii) that $\M_V^e(S,S')\neq\emptyset$.\qed
\end{proof}

\begin{definition}\label{def:MUBPP}
  We shall say that $S'$ is a \emph{Marginal Utility-Based Price Process} (MUBPP) for the investor
  $(U,\edow)$ on the market $S$ if it satisfies (any one of) the equivalent conditions above.
\end{definition}
\begin{remark}[Pricing Contingent Claims]
  Suppose that $S'$ is the price process of a derivative security or a contingent claim which matures at time $T$. The price of the asset at maturity will be equal to the payoff of the contingent claim. If $S'$ is a MUBPP then it is a fair price process for the contingent claim $B=S_T'$ because the investor will not be tempted to take either a long or a short position in the claim at any time.
\end{remark}

\noindent The next result gives a natural characterization of MUBPP's, thus generalizing Theorem 4.2 of Kallsen and K\"uhn (2005) to the continuous time setting. We let $\optQQ_\edow$ denote the local martingale measure which is obtained by normalizing (to a probability measure) the optimal measure $\widehat\mu_\edow$ in the dual problem with endowment $\edow$.
\begin{theorem}\label{thm:MUBPP}
  A semimartingale $S'$ is a MUBPP if and only if it is a $\optQQ_\edow$ local martingale.
\end{theorem}

\begin{proof}
  If $S'$ is a $\optQQ_\edow$-local martingale then $\optQQ_\edow\in\M_V^a(S,S')$. It suffices to prove that $u_{\edow}(S,S')\le u_{\edow}(S)$. Indeed, for all $(H,H')\in\H^\perm(S,S')$ the wealth process $(H,H')\cdot(S,S')$ is a $\optQQ_\edow$-supermartingale and hence
  \begin{align*}
    \E{U(\edow+((H,H')\cdot(S,S'))_T)}
    &\le\E{V\left(\radnik[\widehat\mu_\edow]\right)+\radnik[\widehat\mu_\edow](\edow+((H,H')\cdot(S,S'))_T)}\\
    &\le v_\edow(S)=u_{\edow}(S).
  \end{align*}
  Taking the supremum of the left hand side over all $(H,H')\in\H^\perm(S,S')$ gives the required inequality.

  For the other direction: Since $\M_V^a(S)\supseteq\M_V^a(S,S')$ it follows that $v_{\edow}(S)\le v_{\edow}(S,S')$. Suppose that $u_{\edow}(S,S')=u_{\edow}(S)$. Since there is no duality gap for either of the markets $S$ and $(S,S')$ it follows that $v_{\edow}(S)=v_{\edow}(S,S')$, so in fact $\widehat\mu_\edow$ must be the unique optimizer for $v_{\edow}(S,S')$. Hence $\optQQ_\edow\in\M_V^a(S,S')$. Hence $S'$ is a local martingale under $\optQQ_\edow$.\qed
\end{proof}

\begin{remarks}
  \begin{enumerate}
    \item As a special case when the market is trivial, in the sense that $S=0$, the above theorem reduces to the statement that $S'$ will be an uninteresting investment opportunity (for any investor) if and only if it is a $\PP$-local martingale.
    \item If $\M_V^e(S)$ is non-empty then $\optQQ_\edow\sim\PP$, and any $\optQQ_\edow$-local martingale will be a semimartingale under $\PP$ (see e.g. Theorem II.2 of Protter (2003)). Otherwise this may not be the case, and a $\optQQ_\edow$-local martingale may not be a MUBPP (see Example 4.3 of Kallsen and K\"uhn (2005)).
  \end{enumerate}
\end{remarks}

\section{Utility Indifference Pricing}\label{sec:indiff}

Indifference pricing is currently a highly actively area of research, and, correspondingly, the amount of literature on the topic seems to have grown exponentially in recent years. A recent overview of the existing indifference price literature can be found in Henderson and Hobson (2004).

Within mathematical finance, the first reference to indifference pricing is Hodges and Neuberger (1989) in the context of transaction costs. More recently, Rouge and El Karoui (2000) have studied various aspects of the seller's indifference price for bounded contingent claims, exponential utility and a Brownian filtration. They show that the indifference price tends to the super-replication price as the risk aversion parameter tends to infinity, and to the minimal entropy price as the risk aversion parameter tends to $0$. The minimal entropy price is related to Davis's ``fair price'', as mentioned already in Section \ref{sec:mubpp}.

In the semimartingale setting Delbaen et al. (2002) and Becherer (2003) study the indifference price for exponential utility. Between them, these papers generalize the asymptotical analysis of Rouge and El Karoui (2000) with respect to the risk aversion parameter. Recently, Mania and Schweizer (2005) have studied a dynamic version of indifference prices for exponential utility, and Kl\"oppel and Schweizer (2006) consider dynamic versions of utility indifference prices via convex risk measures.

In Propositions \ref{thm:indifprice} and \ref{thm:avindifprice} we present new results on indifference prices. These include continuity properties for the indifference price and volume asymptotics of the average indifference price for the case of a general utility function, unbounded endowment and unbounded contingent claims.

\vskip\baselineskip

Let us consider the point of view of an investor with a utility function $U$ and a random endowment $\edow$, who is considering buying a contingent claim $B$. Assumptions \ref{ass:rae}, \ref{ass:noarb} and \ref{ass:newweakerendowment} will hold throughout this section. In addition, all contingent claims are assumed to satisfy the conditions in Assumption \ref{ass:newweakerendowment}.
\begin{definition}
  The \emph{utility indifferent purchase (bid) price}, $p=p_\edow(B)=p(B;U,\edow)$, of $B$ is defined implicitly as the solution to the equation
  \begin{equation}\label{eqn:primalindif}
    u_{\edow+B-p}=u_\edow.
  \end{equation}
\end{definition}

Our next result concerns the existence and uniqueness of a solution to equation \eqref{eqn:primalindif}, along with various fundamental properties of the indifference price. As usual, $\optQQ_\edow$ denotes the local martingale measure which is obtained by normalizing the optimal measure $\widehat\mu_\edow$ in the dual problem with endowment $\edow$. The expectation $\E[\optQQ_\edow]B$ is related to Davis's fair price (see Davis (1997)).

\begin{proposition}[Indifference Prices]\label{thm:indifprice}
  For any contingent claim satisfying Assumption \ref{ass:newweakerendowment}, there exists a unique solution, $p$, to equation \eqref{eqn:primalindif}. The utility indifferent purchase price, $p_\edow(B)=p$, is therefore well defined. Moreover,
  \begin{enumerate}
    \item (Range of prices) \[ \inf_{\QQ\in\M_V^a}\E[\QQ]B\le p(B)\le \E[\optQQ_\edow]B; \]

    \item (Translation invariance) For $c\in\RR$ we have $p_\edow(B+c)=p_\edow(B)+c$;

    \item (Pricing replicable claims) If $B=(H\cdot S)_T$ for some $H\in\H^\mg$ then $p_\edow(B)=0$;

    \item (Monotonicity) If $B\le C$ then $p_\edow(B)\le p_\edow(C)$. Moreover, if $\M_V^e\neq\emptyset$ then $p_\edow(B)<p_\edow(C)$ for any $B\le C$ such that $B\neq C$;

    \item (Concavity) Given contingent claims $B_1,B_2$ and $\lambda\in[0,1]$,
    \[ p_\edow\big(\lambda B_1+(1-\lambda)B_2\big)\ge\lambda p_\edow(B_1)+(1-\lambda)p_\edow(B_2); \]

    \item (Pricing via entropic penalty)
    \[ p_\edow(B)=\inf_{\QQ\in\M_V^a}\{\E[\QQ]B+\alpha(\QQ)\}, \]
    where the penalty functional, or dual Orlicz norm, $\alpha:\M_V^a\rightarrow[0,\infty)$ is defined by
    \[ \alpha(\QQ)\coloneqq\inf_{y>0}\frac1y\left\{\E{V\left(y\radnik\right)+y\radnik\edow}-v_\edow\right\}; \]

    \item (Strong Continuity) If $(B_n)_{n\in\NN}$ is a sequence of contingent claims such that
    \[ \sup_{\QQ\in\M_V^a}\E[\QQ]{B_n-B}\rightarrow0\text{ and }\inf_{\QQ\in\M_V^a}\E[\QQ]{B_n-B}\rightarrow0 \]
    then $p_\edow(B_n)\rightarrow p_\edow(B)$;

    \item (Fatou property) If $(B_n)_{n\ge0}$ is a sequence of contingent claims which uniformly satisfy Assumption \ref{ass:newweakerendowment} (in the sense that the sub- and super-replicating strategies can be chosen independently of $n$) then
    \[ p_\edow(\limsup\nolimits_nB_n)\ge\limsup\nolimits_np_\edow(B_n); \]

    \item (Continuity from above) If $(B_n)_{n\in\NN}$ is a sequence of contingent claims such that $B_n\searrow B$ $\PP$-a.s. then $p_\edow(B_n)\searrow p_\edow(B)$.
  \end{enumerate}
\end{proposition}

\begin{proof}
  Let $u_\edow(B)$ be as defined in Section \ref{sec:dependow}. The existence and uniqueness of a solution to \eqref{eqn:primalindif}, as well as statement (i) follow immediately from Corollary \ref{thm:one}. Part (ii) is obvious, part (iii) is a trivial consequence of the fact that the vector space $\H^\mg$ lies within $\H^\perm$, so $\pm H+\H^\perm=\H^\perm$. Part (iv) follows from monotonicity of $u_\edow(.)$. Part (v) also follows from Corollary \ref{thm:one}: By concavity of  $u_\edow(.)$,
  \begin{align*}
    &u_\edow\big(\lambda B_1+(1-\lambda)B_2-\lambda p_\edow(B_1)-(1-\lambda)p_\edow(B_2)\big)\\
    &\qquad\qquad\qquad\ge \lambda u_\edow\big(B_1-p_\edow(B_1)\big)+(1-\lambda)
    u_\edow\big(B_2-p_\edow(B_2)\big) =u_\edow.
  \end{align*}
  By monotonicity of $u_\edow(.)$ we have $p_\edow(\lambda B_1+(1-\lambda)B_2)\ge\lambda p_\edow(B_1)+(1-\lambda)p_\edow(B_2)$.

  (vi) Let $\widehat\mu$ be the optimal measure in the dual problem for the endowment $\edow$. Then by monotonicity and continuity of $u_\edow(.)$,
  \begin{align*}
    p_\edow(B)&=\inf\{p:u_\edow(B-p)<v_\edow\}\\
    &=\inf\left\{p:\inf_{\mu\in\Cone(\M_V^a)}\left\{\E{V\left(\radnik[\mu]\right)+\radnik[\mu](\edow+B-p)}-v_\edow\right\}<0\right\}\\
    &=\inf\left\{p:\inf_{\QQ\in\M_V^a}\inf_{y>0}\left\{\E{V\left(y\radnik\right)+y\radnik\edow}+y\E[\QQ]B-v_\edow-yp\right\}<0\right\}\\
    &=\inf\left\{p:\inf_{\QQ\in\M_V^a}\left\{\E[\QQ]B+\inf_{y>0}\left\{\frac1y\left(\E{V\left(y\radnik\right)+y\radnik\edow}-v_\edow\right)\right\}\right\}<p\right\}\\
    &=\inf_{\QQ\in\M_V^a}\{\E[\QQ]B+\alpha(\QQ)\}.
  \end{align*}
  (vii) This follows from part (vi) because
  \begin{align*}
    -\sup_{\QQ\in\M_V^a}\E[\QQ]{B-B_n}
    & \le\inf_{\QQ\in\M_V^a}\{\E[\QQ]{B_n}+\alpha(\QQ)\}-\inf_{\QQ\in\M_V^a}\{\E[\QQ]B+\alpha(\QQ)\} \\
    & \le \sup_{\QQ\in\M_V^a}\E[\QQ]{B_n-B}.
  \end{align*}
  Hence $|p_\edow(B_n)-p_\edow(B)|\le\sup_{\QQ\in\M_V^a}|\E[\QQ]{B_n-B}|$.

  (viii) It follows immediately from the conditions that $\limsup_n B_n$ also satisfies Assumption \ref{ass:newweakerendowment}. Moreover, using Fatou's Lemma,
  \begin{align*}
    p_\edow(\limsup\nolimits_nB_n) &= \inf_{\QQ\in\M_V^a}\{\E[\QQ]{\limsup\nolimits_nB_n}+\alpha(\QQ)\} \\
    &\ge \inf_{\QQ\in\M_V^a}\limsup\nolimits_n\{\E[\QQ]{B_n}+\alpha(\QQ)\} \\
    &\ge \limsup\nolimits_n\inf_{\QQ\in\M_V^a}\{\E[\QQ]{B_n}+\alpha(\QQ)\} \\
    &= \limsup\nolimits_np_\edow(B_n).
  \end{align*}

  (ix) This follows immediately from (viii).\qed
\end{proof}

\begin{remarks}
  \begin{enumerate}
    \item The quantity $\pi_\edow(B)\coloneqq-p_\edow(-B)$ is the utility indifferent sale (offer) price of $B$. It follows immediately from concavity of $p_\edow$ that $\pi_\edow(B)\ge p_\edow(B)$;

    \item The certainty equivalent $c=c_\edow(B)=c(B;U,\edow)$ of a contingent claim $B$ is defined implicitly as the solution to the equation $u_{\edow+B}=u_{\edow+c}$. Clearly, $p_\edow(B)=-c_{\edow+B}(-B)$. This alternative, but intimately related, approach to pricing contingent claims has been investigated by Frittelli (2000) and Becherer (2003) with zero random endowment.

    \item The buyer's indifference price has an immediate interpretation for a risk manager with a risky position $B$, as it is the largest amount that another investor (with utility function $U$ and endowment $\edow$) would be prepared to pay in order to take on the risk inherent in $B$. Put differently, $\rho(B;U,\edow)\coloneqq-p(B;U,\edow)$ is the smallest amount of money that would have to be added to $B$ in order to make it acceptable for the risk to be taken on by another investor, and therefore has the properties of a convex risk measure. We refer the reader to F\"ollmer and Schied (2002) for details about convex risk measures. It is no coincidence that the penalty functional $\alpha$ in Proposition \ref{thm:indifprice} is a type of dual Orlicz norm similar to those arising from considering risk measures associated to shortfall risk, as in Section 3 of F\"ollmer and Schied (2002).
  \end{enumerate}
\end{remarks}

In the next result we analyze volume asymptotic properties of the average indifference price, generalizing results of Delbaen et al. (2002) and Becherer (2003) to any utility function satisfying Assumption \ref{ass:rae}. Related results on asymptotics (all for exponential utility) have been obtained by Rouge and El Karoui (2000), Fujiwara and Miyahara (2003) and Stricker (2004).

\begin{definition}
  For $\beta>0$, the \emph{average utility indifferent purchase price} for $\beta$ units of the contingent claim $B$ is defined by
  \[ p_\edow(B,\beta)\coloneqq\frac{p_\edow(\beta B)}\beta. \]
\end{definition}

\begin{proposition}[Volume Asymptotics]\label{thm:avindifprice}
  Suppose that $B$ satisfies the conditions of Assumption \ref{ass:newweakerendowment}. Then
  $p_\edow(B,\beta)$ is a continuous, non-increasing function of $\beta$. Moreover,
  \begin{enumerate}
    \item $\displaystyle\inf_{\QQ\in\M_V^a}\E[\QQ]B\le p(B,\beta)\le\E[\optQQ_\edow]B$;
    \item $\displaystyle\lim_{\beta\rightarrow\infty}p_\edow(B,\beta)=\inf_{\QQ\in\M_V^a}\E[\QQ]B$;
    \item $\displaystyle\lim_{\beta\rightarrow0}p_\edow(B,\beta)=\E[\optQQ_\edow]B$.
  \end{enumerate}
\end{proposition}

\begin{proof}
  It follows immediately from Proposition \ref{thm:indifprice} that $p_\edow(B,\beta)$ is a continuous function of $\beta$. Take $0<\beta_1\le\beta_2$. Then by concavity of $p_\edow$ (setting $\lambda=\beta_1/\beta_2$, $B_1=\beta_2B$ and $B_2=0$),
  \begin{equation*}
    p_\edow(B,\beta_1)=\frac1{\beta_1}p\left(\frac{\beta_1}{\beta_2}\beta_2B\right) \ge \frac{1}{\beta_2}p_\edow(\beta_2B)+\left(\frac1{\beta_1}-\frac1{\beta_2}\right)p_\edow(0) = p_\edow(B,\beta_2).
  \end{equation*}

  (i) This follows immediately from Proposition \ref{thm:indifprice}.

  (ii) Suppose for a contradiction there exists a $\QQ\in\M_V^a$ such that $\E[\QQ]B<\lim_{\beta\rightarrow\infty}p_\edow(B,\beta)$. Then for any $\beta>0$,
  \begin{align*}
    u_\edow&=u_\edow(\beta B-p_\edow(\beta B))\\
    &=\inf_{\mu\in\Cone(\M_V^a)}\E{V\left(\radnik[\mu]\right)+\radnik[\mu](\edow+\beta(B-p_\edow(B,\beta)))}\\
    &\le \E{V\left(\radnik\right)+\radnik\edow}+\beta(\E[\QQ]B-p_\edow(B,\beta)).
  \end{align*}
  Taking the limit as $\beta\rightarrow\infty$, the right hand side tends to $-\infty$, which is a contradiction.

  (iii) Applying Proposition \ref{thm:indifprice}(vii), we see that $p_\edow(B/n)\rightarrow0$ as  $n\rightarrow\infty$. Hence, defining $B_n:=B/n-p_\edow(B/n)$ we see that
  \[ \sup_{\QQ\in\M_V^a}\E[\QQ]{B_n}\rightarrow0\qquad\text{and}\qquad\inf_{\QQ\in\M_V^a}\E[\QQ]{B_n}\rightarrow0 \]
  as $n\rightarrow\infty$. Let $\widehat\mu_n$ (resp. $\widehat\QQ_n$) denote the optimal measure (resp. normalized measure) in the dual problem with endowment $\edow+B_n$. By Proposition \ref{thm:deponendow}(iii), $\widehat\mu_n\rightarrow\widehat\mu_\edow$ weakly in total variation as $n\rightarrow\infty$. Since $\{\widehat\mu_n:n\in\NN\}\subseteq\ba_+^\sigma$ it follows that $\widehat\QQ_n\rightarrow\widehat\QQ_\edow$ weakly in total variation as $n\rightarrow\infty$. From the definition of an indifference price,
  \begin{align*}
    u_\edow &= u_\edow(B_n)=\E{V\left(\radnik[\widehat\mu_n]\right)+\radnik[\widehat\mu_n](\edow+B_n)}\\
    &\ge u_\edow+\E{\radnik[\widehat\mu_n](B-np_\edow(B/n))}\big/n
    =u_\edow+\|\widehat\mu_n\|(\E[\widehat\QQ_n]{B}-p_\edow(B,1/n))/n.
  \end{align*}
  Rearranging gives $p_\edow(B,1/n)\ge\E[\optQQ_n]B$. Since $B$ satisfies Assumption \ref{ass:newweakerendowment}, we may find an $H'\in\H^\mg$ such that $\widetilde B:=B-(H'\cdot S)_T$ is bounded below. Now,
  \begin{align*}
    \lim_{n\rightarrow\infty}p_\edow(B,1/n)&\ge\liminf_{n\rightarrow\infty}\E[\optQQ_n]B
    =\liminf_{n\rightarrow\infty}\E[\optQQ_n]{\widetilde B}
    =\liminf_{n\rightarrow\infty}\sup_m\E[\optQQ_n]{\widetilde B\wedge m}\\
    &\ge\sup_m\liminf_{n\rightarrow\infty}\E[\optQQ_n]{\widetilde B\wedge m}=
    \sup_m\E[\optQQ_\edow]{\widetilde B\wedge m}=\E[\optQQ_\edow]{\widetilde
    B}=\E[\optQQ_\edow]B.
  \end{align*}
  \qed
\end{proof}

\section{Appendix}

\begin{proof}[ of Proposition \ref{thm:deponendow}]
  Note first that by Theorem \ref{thm:main},
  \[ u(\edow)=u_\edow=v_\edow=\min_{\mu\in\Cone(\M_V^a)}\{H(\mu)+\mu(\edow)\}, \]
  where
  \begin{equation}\label{eqn:shiftentropy}
    H(\mu):=\E{V\left(\radnik[\mu]\right)}\qquad\text{and}\qquad\mu(\edow) :=\E{\radnik[\mu]\edow}.
  \end{equation}
  (i) Take $\edow_1\le \edow_2$. Let $\widehat\mu_{\edow_2}$ be optimal in the dual problem for the endowment $\edow_2$. Then
  \begin{align}
    u(\edow_1)&=\min_{\mu\in\Cone(\M_V^a)}\{H(\mu)+\mu(\edow_1)\}\le H(\widehat\mu_{\edow_2})+\widehat\mu_{\edow_2}(\edow_1) \notag\\
    &\le H(\widehat\mu_{\edow_2})+\widehat\mu_{\edow_2}(\edow_2)\label{eqn:one} =
    \min_{\mu\in\Cone(\M_V^a)}\{H(\mu)+\mu(\edow_2)\}= u(\edow_2)
  \end{align}
  If $\edow_2=\edow_1+\epsilon$ then the inequality \eqref{eqn:one} is strict. Take $\edow_1\le\edow_2$, and $\edow_1\neq \edow_2$. If $\M_V^e\neq\emptyset$ then $\widehat\mu\sim\PP$, so again, \eqref{eqn:one} is strict.

  (ii)
  \begin{align*}
    u(\lambda \edow_1+(1-\lambda)\edow_2) &= \min_{\mu\in\Cone(\M_V^a)}\{H(\mu)+\mu(\lambda \edow_1+(1-\lambda)\edow_2)\}\\
    &= \min_{\mu\in\Cone(\M_V^a)}\{\lambda(H(\mu)+\mu(\edow_1))+(1-\lambda)(H(\mu)+\mu(\edow_2))\}\\
    &\ge \lambda u(\edow_1)+(1-\lambda) u(\edow_2).
  \end{align*}
  (iii) Define $\edow_\infty=\edow$. First note that due to Assumption \ref{ass:newweakerendowment}, and the conditions \eqref{eqn:surprise}
  \begin{equation*}
    C\coloneqq\inf_{n\in\NN\cup\{\infty\}}\inf_{\QQ\in\M_V^a}\E[\QQ]{\edow_n} \ge
    \inf_{\QQ\in\M_V^a}\E[\QQ]{\edow}+\inf_{n\in\NN\cup\{\infty\}}\inf_{\QQ\in\M_V^a}\E[\QQ]{\edow_n-\edow}\ge-\infty.
  \end{equation*}
  Moreover, for $\QQ'\in\M_V^a$ fixed we have
  \[ C'\coloneqq\sup_{n\in\NN\cup\{\infty\}}\{H(\QQ')+\QQ'(\edow_n)\} \le H(\QQ')+\QQ'(\edow)+\sup_{n\in\NN\cup\{\infty\}}\sup_{\QQ\in\M_V^a}\E[\QQ]{\edow_n-\edow} < \infty. \]
  Since $V$ is convex and $V'(y)\rightarrow\infty$ there exists a constant $r<\infty$ such that $V(y)+Cy>C'$ for all $y\ge r$. If $\mu\in\Cone(\M_V^a)$ satisfies $\|\mu\|\ge r$ then by Jensen's inequality, for any $n\in\NN$,
  \[ H(\mu)+\mu(\edow_n)\ge V(\|\mu\|)+C\|\mu\|>C'\ge H(\QQ')+\QQ'(\edow_n). \]
  Thus $\mu$ cannot be the minimizer in the dual problem \eqref{eqn:maindual} with endowment $\edow_n$ for any $n\in\NN\cup\{\infty\}$. It follows therefore that
  \begin{align*}
    u(\edow_n)-u(\edow)&=\inf_{\mu\in\Cone(\M_V^a)\cap B(r)}\{H(\mu)+\mu(\edow_n)\}-\inf_{\mu\in\Cone(\M_V^a)\cap B(r)}\{H(\mu)+\mu(\edow)\}\\
    &\le \sup_{\mu\in\Cone(\M_V^a)\cap B(r)}\mu(\edow_n-\edow) \le
    \max\{r\sup_{\QQ\in\M_V^a}\E[\QQ]{\edow_n-\edow},0\},
  \end{align*}
  where $B(r)$ denotes the ball of radius $r$ in $(\ba,\|.\|)$. Similarly,
  \[ u(\edow_n)-u(\edow)\ge\min\{r\inf_{\QQ\in\M_V^a}\E[\QQ]{\edow_n-\edow},0\}.\]
  Thus $|u(\edow_n)-u(\edow)|\rightarrow0$ as $n\rightarrow\infty$.

  For ease of notation in the proof of the second statement, define $\widehat\mu_n:=\widehat\mu_{\edow_n}$. Then
  \begin{align*}
    |H(\widehat\mu_n)+\widehat\mu_n(\edow)-v_{\edow}|
      &= |H(\widehat\mu_n)+\widehat\mu_n(\edow_n)-v_{\edow}-\widehat\mu_n(\edow_n-\edow)|\\
      &\le |v_{\edow_n}-v_{\edow}|+|\widehat\mu_n(\edow_n-\edow)|\\
      &\le |u(\edow_n)-u(\edow)|+r\sup_{\QQ\in\M_V^a}|\E[\QQ]{\edow_n-\edow}|.
  \end{align*}
  Hence $(\widehat\mu_n)_{n\ge0}$ is a minimizing sequence for the dual problem with endowment $\edow$.

  We may assume without loss of generality that $\edow$ satisfies Assumption \ref{ass:edow}; in the case when $\edow$ satisfies the weaker Assumption \ref{ass:newweakerendowment}, we may find an $H'\in\H^\mg$ such that $\widetilde \edow\coloneqq \edow-(H'\cdot S)_T$ satisfies Assumption \ref{ass:edow}, and note that $\widehat\mu_n$ is a minimizing sequence for the dual problem with endowment $\widetilde\edow$.

  We now show that $\widehat\mu_n\rightarrow\widehat\mu_\edow$ weakly in total variation as $n\rightarrow\infty$. Suppose for a contradiction that this is not the case. Then there exists a $\sigma(\ba,L^\infty)$-open neighbourhood, $U$, of $\widehat\mu_\edow$ such that for all $N\in\NN$ there exists an $n\ge N$ such that $\widehat\mu_n\not\in U$. We may therefore find a subsequence $\{\widehat\mu_{n_k}\}_{k\ge0}$ which lies in the closed set $\ba\setminus U$. Since this subsequence lies inside the weak* compact ball of radius $r$ in $\ba$, it has a cluster point. There exists, therefore, a subnet $\{\widehat\mu_\alpha\}_{\alpha\in A}$ of $\{\widehat\mu_{n_k}\}_{k\ge0}$ which converges to some $\mu\in\ba\setminus U$. Since $\{\VV_\edow(\widehat\mu_{n_k})\}_{k\ge0}$ converges to $v_\edow$, the subnet $\{\VV_\edow(\widehat\mu_\alpha)\}_{\alpha\in A}$ also converges to $v_\edow$. Since $\VV_\edow$ is lower semicontinuous in the weak* topology $\sigma(\ba,L^\infty)$ we have
  \[ v_\edow=\lim_{\alpha}\VV_\edow(\widehat\mu_\alpha)\ge\VV_\edow(\mu). \]
  Since $\widehat\mu_\edow$ is the unique minimizer in the dual problem, we must have $\mu=\widehat\mu_\edow\in U$, which is the required contradiction.

  (iv) Suppose that $\{\edow_n\}_{n\ge0}$ uniformly satisfy Assumption \ref{ass:newweakerendowment}. Then the contingent claim $\limsup_n \edow_n$ also satisfies Assumption \ref{ass:newweakerendowment}. Moreover, using Fatou's Lemma,
  \begin{align}
    u(\limsup\nolimits_n\edow_n) &= \inf_{\mu\in\Cone(\M_V^a)}\{H(\mu)+\mu(\limsup\nolimits_n\edow_n)\} \notag\\
    &\ge \limsup\nolimits_n\inf_{\mu\in\Cone(\M_V^a)}\{H(\mu)+\mu(\edow_n)\} \notag\\
    &= \limsup\nolimits_nu(\edow_n).\label{eqn:fatou1}
  \end{align}
  By our assumptions, the contingent claim $\liminf_n\edow_n$ also satisfies Assumption \ref{ass:newweakerendowment}, and we may find an $H'\in\H^\mg$ such that $\widetilde\edow_n:=\edow_n-(H'\cdot S)_T$ satisfies Assumption \ref{ass:edow}. Using Fatou's Lemma, and Theorem \ref{thm:minimalassumption}(i),
  \begin{align}
    u(\liminf\nolimits_n\edow_n) &= \sup_{H\in\H^\perm}\E{U\big((H\cdot S)_T+\liminf\nolimits_n\edow_n\big)} \notag\\
    &= \sup_{H\in\H^\perm}\E{U\big((H\cdot S)_T+\liminf\nolimits_n\widetilde \edow_n\big)} \notag\\
    &= \sup_{H\in\H^\adm}\E{\liminf\nolimits_nU\big((H\cdot S)_T+\widetilde \edow_n\big)} \notag\\
    &\le \sup_{H\in\H^\adm}\liminf_n\E{U\big((H\cdot S)_T+\widetilde \edow_n\big)} \notag\\
    &\le \liminf_n\sup_{H\in\H^\adm}\E{U\big((H\cdot S)_T+\widetilde \edow_n\big)} \notag\\
    &= \liminf_n\sup_{H\in\H^\perm}\E{U\big((H\cdot S)_T+\widetilde \edow_n\big)} \notag\\
    &= \liminf_n\sup_{H\in\H^\perm}\E{U\big((H\cdot S)_T+\edow_n\big)} \notag\\
    &= \liminf\nolimits_nu(\edow_n).\label{eqn:fatou2}
  \end{align}
  Finally, if $\edow_n\rightarrow\edow$ $\PP$-a.s. as $n\rightarrow\infty$ then by applying both Fatou properties \eqref{eqn:fatou1} and \eqref{eqn:fatou2} above,
  \begin{equation*}
    \liminf\nolimits_nu(\edow_n)\ge u(\edow)\ge\limsup\nolimits_nu(\edow_n).
  \end{equation*}
  \qed
\end{proof}

\section*{References}

\setlength\parskip\baselineskip

\noindent{\scshape Becherer,~D.} (2003): Rational Hedging and Valuation of Integrated Risks under Constant Absolute Risk Aversion. \emph{Insurance: Mathematics \& Economics} 33, 1--28.

\noindent{\scshape Biagini,~S. {\upshape and} M.~Frittelli} (2005): Utility Maximization in Incomplete Markets for Unbounded Processes. \emph{Finance Stochast.} 9, 493--517.

\noindent{\scshape Biagini,~S. {\upshape and} M.~Frittelli} (2007): The Supermartingale Property of the Optimal Wealth Process for General Semimartingales. \emph{Finance Stochast.} 11, 253--266.

\noindent{\scshape Biagini,~S. {\upshape and} M.~Frittelli} (2006): A Unified Framework for Utility Maximization Problems: An Orlicz Space Approach. \emph{Preprint}.

\noindent{\scshape Cvitani\'c,~J., W.~Schachermayer {\upshape and} H.~Wang} (2001): Utility Maximization in Incomplete Markets with Random Endowment. \emph{Finance Stochast.} 5, 259--272.

\noindent{\scshape Davis,~M.~H.~A.} (1997): Option Pricing in Incomplete Markets, in \emph{Mathematics of Derivative Securities}, eds. M.~A.~H.~Dempster and S.~R.~Pliska. Cambridge University Press, 216--226.

\noindent{\scshape Delbaen,~F., P.~Grandits, T.~Rheinl\"ander, D.~Samperi, M.~Schweizer {\upshape and}  C.~Stricker} (2002): Exponential Hedging and Entropic Penalties. \emph{Math. Finance} 12, 99--123.

\noindent{\scshape Delbaen,~F. {\upshape and} W.~Schachermayer} (1994): A General Version of the Fundamental Theorem of Asset Pricing. \emph{Mathematische Annalen} 300, 463--520.

\noindent{\scshape Delbaen,~F. {\upshape and} W.~Schachermayer} (1998): The Fundamental Theorem of Asset Pricing for Unbounded Stochastic Processes. \emph{Mathematische Annalen} 312, 215--250.

\noindent{\scshape Delbaen,~F. {\upshape and} W.~Schachermayer} (1999): A Compactness Principle for Bounded Sequences of Martingales with Applications. \emph{Progress in Probability} 45, 137--173.

\noindent{\scshape Delbaen,~F. {\upshape and} W.~Schachermayer} (2006): \emph{The Mathematics of Arbitrage}. Springer Finance.

\noindent{\scshape Dunford,~N. {\upshape and} J.~T.~Schwartz} (1964): \emph{Linear Operators. Part I: General Theory}. Wiley.

\noindent{\scshape Foldes,~L.} (2000): Valuation and Martingale Properties of Shadow Prices: An Exposition. \emph{J. Econ. Dynam. Control} 24, 1641--1701.

\noindent{\scshape F\"ollmer,~H. {\upshape and} A.~Schied} (2002): Convex Measures of Risk and Trading Constraints. \emph{Finance Stochast.} 6, 429--447.

\noindent{\scshape Frittelli,~M.} (2000): Introduction to a Theory of Value Coherent with the No-Arbitrage Principle. \emph{Finance Stochast.} 4, 275--297.

\noindent{\scshape Fujiwara,~T. {\upshape and} Y.~Miyahara} (2003): The Minimal Entropy Martingale Measures for Geometric L\'evy Processes. \emph{Finance Stoch.} 7, 509--531.

\noindent{\scshape Henderson,~V. {\upshape and} D.~Hobson} (2004): Utility Indifference Pricing - An Overview, in \emph{Volume on Indifference Pricing}, ed. R.~Carmona. Princeton University Press, to appear.

\noindent{\scshape Hewitt,~E. {\upshape and} K.~Stromberg} (1965): \emph{Real and Abstract Analysis}. Springer.

\noindent{\scshape Hodges,~S.~D. {\upshape and} A.~Neuberger,} (1989): Optimal Replication of Contingent Claims Under Transaction Costs. \emph{Rev. Futures Markets} 8, 222--239.

\noindent{\scshape Hugonnier,~J. {\upshape and} D.~Kramkov} (2004): Optimal Investment with Random Endowments in Incomplete Markets. \emph{Ann. Appl. Prob.} 14, 845--864.

\noindent{\scshape Hugonnier,~J., D.~Kramkov {\upshape and} W.~Schachermayer} (2005): On Utility-Based Pricing of Contingent Claims in Incomplete Markets. \emph{Math. Finance} 15, 203--212.

\noindent{\scshape Jouini,~E., W.~Schachermayer {\upshape and} N.~Touzi} (2006): Law Invariant Risk Measures have the Fatou Property. \emph{Adv. Math. Econ.} 9, 49--71.

\noindent{\scshape Y.~M. Kabanov {\upshape and} C.~Stricker} (2002): On the Optimal Portfolio for the Exponential Utility Maximization: Remarks to the Six-Author Paper. \emph{Math. Finance} 12, 125--134.

\noindent{\scshape Kallsen,~J.} (2002): Derivative Pricing based on Local Utility Maximization. \emph{Fin. Stoch.} 6, 115--140.

\noindent{\scshape Kallsen,~J. {\upshape and} C.~K\"uhn} (2004): Pricing Derivatives of American and Game Type in Incomplete Markets. \emph{Fin. Stoch.} 8, 261--284.

\noindent{\scshape Kallsen,~J. {\upshape and} C.~K\"uhn} (2005): On Utility-Based Derivative Pricing with and without Intermediate Trades. \emph{Statistics and Decisions}, to appear.

\noindent{\scshape Karatzas,~I. {\upshape and} S.~G. Kou} (1996): On the pricing of contingent claims under constraints. \emph{Ann. Appl. Prob.} 6, 321--369.

\noindent{\scshape Kl\"oppel,~S. {\upshape and} M.~Schweizer} (2006): Dynamic Utility Indifference Valuation via Convex Risk Measures. \emph{Math. Finance}, to appear.

\noindent{\scshape Kramkov,~D. {\upshape and} W.~Schachermayer} (1999): The Asymptotic Elasticity of Utility Functions and Optimal Investment in Incomplete Markets. \emph{Ann. Appl. Prob.} 9, 904--950.

\noindent{\scshape Luenberger,~D.~G.} (1969): \emph{Optimization by Vector Space Methods}. Wiley.

\noindent{\scshape Mania,~M. {\upshape and} M.~Schweizer} (2005): Dynamic exponential utility indifference valuation. \emph{Ann. Appl. Prob.} 15, 2113--2143.

\noindent{\scshape Owen,~M.~P.} (2002): Utility Based Optimal Hedging in Incomplete Markets. \emph{Ann. Appl. Prob.} 12, 691--709.

\noindent{\scshape Protter,~P.} (2003): \emph{Stochastic Integration and Differential Equations, a New Approach}. Springer.

\noindent{\scshape R.~Rouge {\upshape and} N.~El Karoui} (2000): Pricing via Utility Maximization and Entropy. \emph{Math. Finance} 10, 259--276.

\noindent{\scshape Schachermayer,~W.} (2001): Optimal Investment in Incomplete Markets when Wealth may become Negative. \emph{Ann. Appl. Prob.} 11, 694--734.

\noindent{\scshape Schachermayer,~W.} (2003): A Super-Martingale Property of the Optimal Portfolio Process. \emph{Finance Stochast.} 7, 433--456.

\noindent{\scshape Shiryaev,~A.~N.} (1995): \emph{Probability}, second ed. Springer.

\noindent{\scshape Stricker,~C.} (2004): Indifference Pricing with Exponential Utility, Seminar on Stochastic Analysis, \emph{Random Fields and Applications IV}, eds. R.~Dalang, M.~Dozzi, and F.~Russo. Boston: Birkh\"auser, 325--330.

\noindent{\scshape \v{Z}itkovi\'c,~G.} (2005): Utility Maximization with a Stochastic Clock and an Unbounded Random
Endowment. \emph{Ann. Appl. Prob.} 15, 748--777.

\end{document}